\documentclass[number,sort&compress]{elsarticle}

\usepackage[english]{babel}
\usepackage[utf8]{inputenc}\usepackage[T1]{fontenc}
\usepackage{amsfonts}
\usepackage{amsmath, amsthm, amssymb}
\usepackage{mathtools}
\usepackage{booktabs}  
\usepackage{gensymb}  
\usepackage{siunitx}
\usepackage{multirow}
\usepackage{rotating}

\usepackage{tikz}
\usetikzlibrary{external, patterns, shapes.misc, backgrounds, pgfplots.groupplots}
\usepackage{tikz-dimline}
\tikzset{external/only named=true}

\usepackage{pgfplots}
\pgfplotsset{compat=1.9}
\pgfdeclarelayer{background}
\pgfsetlayers{background,main} 
\pgfplotsset{select coords between index/.style 2 args={
    x filter/.code={
        \ifnum\coordindex<#1\fi
        \ifnum\coordindex>#2\fi
    }
}}

\usepackage{csvsimple}
\usepackage{subcaption}
 
\usepackage{hyperref}

\usepackage[noabbrev,capitalize]{cleveref}
\crefname{equation}{}{}
\crefname{remark}{remark}{remarks}
\Crefname{remark}{Remark}{Remarks}

\newtheorem{theorem}{Theorem}
\theoremstyle{plain}

\theoremstyle{plain}

\newtheorem{definition}[theorem]{Definition}
\theoremstyle{plain}

\newtheorem{remark}[theorem]{Remark}
\theoremstyle{plain}

\theoremstyle{plain}

\newcommand{\dx}{\,dx}
\newcommand{\ds}{\,ds}
\newcommand{\hp}{s}

\let\thetaori\theta

\newcommand{\Rey}{\operatorname{\mbox{\textit{Re}}}}
\newcommand{\Eo}{\operatorname{\mbox{\textit{Eo}}}}
\newcommand{\Ca}{\operatorname{\mbox{\textit{Ca}}}}
\newcommand{\At}{\operatorname{\mbox{\textit{At}}}}
\newcommand{\Cn}{\operatorname{\mbox{\textit{Cn}}}}
\newcommand{\Pe}{\operatorname{\mbox{\textit{Pe}}}}

\journal{Journal of Computational Physics}

\begin{document}

\begin{frontmatter}
\title{Comparison of Energy Stable Simulation of Moving Contact Line Problems using a
Thermodynamically Consistent Cahn--Hilliard Navier--Stokes Model\tnoteref{grants}
}

\tnotetext[grants]{The first and third author acknowledge the North-German Supercomputing Alliance (HLRN) for providing HPC resources 
that have contributed to the research results reported in this paper and 
thank the German Research Foundation (DFG) for the financial support within the project RE 1705/16-1.
The second author gratefully acknowledges the support by the German Research Foundation (DFG) 
through the International Research Training Group
IGDK 1754 ``Optimization and Numerical Analysis for Partial Differential Equations with Nonsmooth Structures".}

\author[TUB]{Henning Bonart \corref{cor}}
\ead{henning.bonart@tu-berlin.de} 
\cortext[cor]{Corresponding author}

\author[TUM]{Christian Kahle}
\ead{christian.kahle@ma.tum.de}

\author[TUB]{Jens-Uwe Repke}
\ead{jens-uwe.repke@tu-berlin.de}

\address[TUB]{Process Dynamics and Operations Group, Technische Universit\"at Berlin,\\ 10623 Berlin, Germany}
\address[TUM]{Center for Mathematical Sciences, Technische Universit\"at M\"unchen,\\ 85748 Garching bei M\"unchen, Germany }

\begin{abstract}
Liquid droplets sliding along solid surfaces are a frequently observed phenomenon in nature, e.g., raindrops on a leaf,
and in everyday situations, e.g., drops of water in a drinking glass. 
To model this situation, we use a phase field approach. 
The bulk model is given by the thermodynamically consistent Cahn--Hilliard~Navier--Stokes model from
[Abels~et~al., Math.~Mod.~Meth.~Appl.~Sc., 22(3), 2012]. 
To model the contact line dynamics we apply the generalized Navier boundary condition for the fluid
and the dynamically advected boundary contact angle condition for the phase field as derived in
[Qian~et~al., J.~Fluid~Mech., 564, 2006].
In recent years several schemes were proposed to solve this model numerically.
While they widely differ in terms of complexity, they all fulfill certain basic
properties when it comes to thermodynamic consistency.
However, an accurate comparison of the influence of the schemes on the moving contact line is rarely found.
Therefore, we thoughtfully compare the quality of the numerical results obtained with three different schemes and two different bulk energy potentials.
Especially, we discuss the influence of the different schemes on the apparent contact angles of a sliding droplet. 
\end{abstract}

 
\end{frontmatter}

\section{Introduction}
\label{sec:I}

Liquid droplets sliding along solid surfaces are a frequently observed phenomenon in nature, e.g., raindrops on a leaf,
and in everyday situations, e.g., drops of water in a drinking glass.
Furthermore, sliding droplets (and consequently the suppression of those) 
are crucial in many industrial applications such as coating or painting 
and separation or reaction processes involving multiple phases and thin liquid films.
The position where the interface between the sliding droplet and the surrounding fluid intersects 
the solid surface is the moving contact line (or contact point if a two dimensional problem is observed).
For details about liquids on surfaces and moving contact lines see the reviews~\cite{Bonn2009, Snoeijer2013} and the references therein.
In a continuum approach, applying the common no-slip boundary condition at the solid surface close to the contact line leads to a non-physical,
logarithmically diverging energy dissipation. 
One possibility to circumvent this difficulty is the coupling of the incompressible Navier--Stokes equations 
with the Cahn--Hilliard equation~\cite{Jacqmin2000}. 
This phase field method models the interface between the fluids with a diffuse interface 
of positive thickness and describes the distribution of the different fluids by a smooth indicator function.
Especially, the Cahn--Hilliard equation allows the contact line to move
naturally on the solid surface due to a diffusive flux across the interface, which is driven by the gradient of the chemical potential. 
Furthermore, the phase field method is able to calculate topological changes like breakup of 
droplets or merging interfaces~\cite{Anderson1998}.
For example in experiments by~\cite{Carlson2012,Eddi2013}, it is found that during
the rapid spreading of a droplet the contact angle can differ from the equilibrium angle given by Young's equation.
To allow for nonequilibrium contact angles, \cite{Jacqmin2000} proposes a
relaxation of the static contact angle boundary condition, see \Cref{ssec:model}, 
and \cite{2006_QianWangShen_Variational_MovingContactLine__BoundaryConditons} 
extends this approach to include the slip at the contact line stemming from the uncompensated Young stress.

In \cite{AbelsGarckeGruen_CHNSmodell} a thermodynamically consistent Cahn--Hilliard~Navier--Stokes phase field model is proposed 
to describe the dynamics of the  two phases in the bulk domain.
It is valid also for different densities of the involved fluids, 
but specific contact line dynamics are not included.
Recently, several numerical schemes for solving this system have been proposed,
see for example,
\cite{
GarckeHinzeKahle_CHNS_AGG_linearStableTimeDisc,
Gruen_Klingbeil_CHNS_AGG_numeric,
GruenMetzger__CHNS_decoupled,
Aland__time_integration_for_diffuse_interface,
Tierra_Splitting_CHNS}.
All these schemes are thermodynamically consistent in the sense, that they mimic
the energy law from \cite{AbelsGarckeGruen_CHNSmodell} in the time discrete or
even in the fully discrete setting. They range from fully coupled and nonlinear
to decoupled and linear, where decoupled means, that the Navier--Stokes and the
Cahn--Hilliard equations are solved sequentially.

These schemes are extended to the Cahn--Hilliard~Navier--Stokes system with 
moving contact lines in various papers.
Here the concepts from the aforementioned papers for the discretization of the
bulk equations are straightforwardly applied. For the case of equal densities, schemes are proposed, e.g., in \cite{2011-HeGlowinskiWang-LeastSquaresCHNSMCL,
2012-GaoWang-gradientStableSchemePhaseFieldMCL,
2015-ShenYangYu-EnergyStableSchemesForCHMCL-Stabilization,
AlandChen__MovingContactLine}
and for the case of different densities in 
\cite{YuYang_MovingContactLine_diffDensities,
GruenMetzger__CHNS_decoupled}.
The model from \cite{AbelsGarckeGruen_CHNSmodell} contains an
additional flux term in the momentum equation, that renders the model thermodynamically consistent.
This term is often neglected, see e.g., \cite{ShenYang_CHNS_DingSpelt_Consistent}.
For the resulting model several discretization schemes are proposed and we refer to the references in
\cite{YuYang_MovingContactLine_diffDensities} for details. 
In all these simulations involving moving
contact lines a polynomial bulk energy potential is applied.
In contrast, we include a double obstacle potential, which is subsequently relaxed, see \cite{HintermuellerHinzeTber}.
In \cite{GarckeHinzeKahle_CHNS_AGG_linearStableTimeDisc,2017_SPP1506_book_AlaHKN__CoparativeSurfactants}
in a numerical benchmark setting the results with this kind of 
energy are typically closer to sharp interface numeric than with the polynomial potential.

To prepare future research on the passive control of droplets sliding on structured or chemically patterned surfaces, we extend the work of
\cite{GarckeHinzeKahle_CHNS_AGG_linearStableTimeDisc} in this paper to the case of moving
contact line dynamics and compare the numerical results with the
corresponding  decoupled scheme from, e.g., \cite{GruenMetzger__CHNS_decoupled} and a
fully linear scheme, that both uses decoupling and stabilization as in
\cite{2015-ShenYangYu-EnergyStableSchemesForCHMCL-Stabilization}. 
We test both the polynomially and the relaxed double-obstacle bulk energy potential,
so that in total we compare six different
combinations of bulk energy potentials and solution schemes.

The remainder of the paper is organized as follows. 
In the second part of the introduction, \Cref{ssec:model}, 
we introduce the continuous model as well as the bulk energy potentials and the contact line energies.
Afterwards, we derive a weak formulation in \Cref{sec:F} and the numerical schemes in \Cref{sec:S}. 
In \Cref{ssec:N:rising_bubble} we compare the different combinations at first in the 
bulk without any contact line. 
Finally, we compare simulation results of sliding droplets on 
inclined surfaces to investigate the accuracy and efficiency of the linearization 
and decoupling strategies as well as the bulk energy potentials
for moving contact line problems in \Cref{ssec:N:sliding_droplet}. 
We conclude our work in \Cref{sec:conclusion}.

\subsection{Model}
\label{ssec:model}
In the fluid domain we consider the thermodynamically consistent model for the simulation of
two-phase flow presented in
\cite{AbelsGarckeGruen_CHNSmodell},
in the variant for nonlinear density
functions proposed in \cite[Eq. 1.10]{AbelsBreit_weakSolution_nonNewtonian_DifferentDensities}.
To model the contact line dynamics we use generalized Navier boundary conditions
for the velocity field together with dynamically advected boundary conditions for the phase field as proposed in 
\cite{2006_QianWangShen_Variational_MovingContactLine__BoundaryConditons}.

In strong form the model reads as follows.
Let $\Omega \subset \mathbb{R}^{d}$ with $d \in\{2,3\}$ denote an open,
polygonally/polyhedrally bounded Lipschitz domain and $I = (0,T]$ with $0<T<\infty$ denote a
time interval. The outer unit normal on $\partial\Omega$ is $\nu_\Omega$. At time $t \in I$
the primal variables are given by the velocity field $v$, the pressure field
$p$, the phase field $\varphi$ and the chemical potential $\mu$.
They satisfy
the following system of equations
\begin{align}
\rho\partial_t v + ((\rho v + J)\cdot\nabla) v + R\frac{v}{2}
-\mbox{div}\left(2\eta Dv\right) + \nabla p &= -\varphi\nabla \mu + \rho g
&& \mbox{ in } \Omega, \label{eq:M:1_NS1}\\
-\mbox{div}(v) &= 0 
&& \mbox{ in } \Omega,\label{eq:M:2_NS2}\\
\partial_t \varphi + v \cdot\nabla \varphi - b\Delta \mu &= 0 
&& \mbox{ in } \Omega,\label{eq:M:3_CH1}\\
-\sigma\epsilon\Delta \varphi + \frac{\sigma}{\epsilon}W^\prime(\varphi) &= \mu 
&& \mbox{ in } \Omega,\label{eq:M:4_CH2}\\
v \cdot \nu_\Omega &= 0 
&& \mbox{ on } \partial\Omega, \label{eq:M:5_NS_BC_1}\\
[2\eta Dv \nu_\Omega + l(\varphi)v_{tan} -  L(\varphi)\nabla \varphi]
\times \nu_\Omega &= 0 
&& \mbox{ on } \partial\Omega, \label{eq:M:6_NS_BC}\\
rB + L(\varphi)&=0 
&& \mbox{ on } \partial\Omega,\label{eq:M:7_CH_BC}\\
\nabla \mu \cdot\nu_\Omega &= 0 
&&\mbox{ on } \partial\Omega,
\label{eq:M:8_mu_neumann}
\end{align}
where we abbreviate $J := -b\frac{\partial\rho}{\partial\varphi}\nabla \mu$,
$R:= -b\nabla\frac{\partial\rho}{\partial\varphi}\cdot\nabla \mu$,
$B := \partial_t\varphi + v \cdot \nabla \varphi$,
$ L:= \sigma\epsilon \nabla\varphi\cdot\nu_\Omega + \gamma^\prime(\varphi)$.
The gravitational acceleration is denoted by $g$ and we abbreviate 
$2Dv := \nabla v + (\nabla v)^t$.
The function $W(\varphi)$ denotes a dimensionless potential
of double-well type, with two strict minima at
$\pm 1$.
 We refer to \Cref{rm:M:freeEnergies} for a discussion of possible choices for $W$.
We formulate \cref{eq:M:1_NS1} with a shifted pressure variable
$p = p^{phys} - \mu\varphi$, where $p^{phys}$ denotes the physical pressure.

The contact line energy is denoted by $\gamma$, see \Cref{rm:M:ContactEnergy}.
The strictly positive, constant parameters for the equations in $\Omega$ are given by the mobility $b >0$, the scaled surface tension $\sigma$, see
\Cref{rm:M:freeEnergies}, and the interfacial thickness parameter $\epsilon$. 
The constant mobility is used for simplicity but the following is also valid for
mobilities that depend on $\varphi$.

The (nonlinear) density function is denoted by $\rho \equiv \rho(\varphi) > 0$ and
satisfies $\rho(-1) = \rho_1$ and $\rho(1) = \rho_2$, with $\rho_2 > \rho_1 > 0$
denoting the constant densities of the two involved fluids. 
The (nonlinear) 
viscosity function is $\eta \equiv \eta(\varphi) > 0$ and satisfies $\eta(-1) =
\eta_1$ and $\eta(1) = \eta_2$, with $\eta_1,\eta_2$ denoting the viscosities of
the involved fluids.

\begin{remark}[Nonlinear density and viscosity]
In general there is no quantitative upper bound available for $\varphi$ and thus 
in particular $|\varphi| > 1$ is possible.
Thus a linear relation between $\varphi$ and $\rho$ can lead to negative densities in practice.
This might appear for especially large density ratios, compare for example \cite[Rem. 4.1]{GruenMetzger__CHNS_decoupled}.
Note, that $\varphi$ can be proven to be bounded in $L^\infty$ if $W^{\prime\prime}$ is uniformly bounded, see e.g. \cite{1995_CaffarelliMulder_LinftyForCahnHilliard}.

It is a common approach to
	cut $\varphi$ when inserting it into the linear function for $\rho$, 
	see e.g. \cite{YuYang_MovingContactLine_diffDensities}. 
	This leads to a nonsmooth relation between $\varphi$ and $\rho$. 
	However, as we require differentiability of $\rho$ to define $R$ and $J$ this is not admissible here.
A second approach is to clip $\rho$ of at some positive value, 	
	see e.g. \cite{GruenMetzger__CHNS_decoupled,Gruen_convergence_stable_scheme_CHNS_AGG}.
	This leads to a uniform bound on $\rho$ based on the Atwood number $\At = \frac{\rho_2 - \rho_1}{\rho_2 + \rho_1}$.
	Here we use the latter approach and define $\rho$ as the following smooth, monotone and strictly positive function
\begin{align}
  \rho(\varphi) = 
  \begin{cases}
\frac{1}{4}\rho_1 &     \mbox{ if } \varphi \leq -\At^{-1},\\
\frac{1}{\rho_1}\left( \frac{\rho_2-\rho_1}{2}\varphi + \frac{\rho_2+\rho_1}{2} \right)^2 + \frac{1}{4}\rho_1
     & 	\mbox{ if } -\At^{-1} < \varphi < -1-\frac{ \rho_1}{\rho_2-\rho_1}, \\
\frac{ \rho_2- \rho_1}{2}\varphi + \frac{\rho_2+\rho_1}{2}
    & \mbox{ if } -1-\frac{ \rho_1}{\rho_2-\rho_1} 
    \leq \varphi \leq 
    1+\frac{\rho_1}{\rho_2-\rho_1},\\ -\frac{1}{\rho_1}\left( \frac{\rho_2-\rho_1}{2}\varphi - \frac{\rho_2+\rho_1}{2} \right)^2 + \rho_2 + \frac{3}{4}\rho_1
      & \mbox{ if } 1+\frac{\rho_1}{\rho_2-\rho_1} < \varphi < \At^{-1},\\
\rho_2 + \frac{3}{4}\rho_1 & \mbox{ if }   \At^{-1} \leq \varphi.
\end{cases}\label{eqn:cutoff_phi}
\end{align} 
For a discussion we refer to \cite[Rem. 2.1]{Gruen_convergence_stable_scheme_CHNS_AGG}.
The nonlinear viscosity $\eta(\varphi)$ can be defined analogously. 

We note, that the total mass $\int_\Omega \rho(\varphi)\dx$ is only conserved if $\rho(\varphi)$ is a linear 
function on the (a-priori unknown) image of $\varphi$, 
see e.g., \cite[Rem. 1]{GarckeHinzeKahle_CHNS_AGG_linearStableTimeDisc}, 
while $\int_\Omega \varphi\dx$ is a conserved quantity, 

\end{remark}

As boundary data we use generalized Navier boundary conditions for
the velocity field and dynamically advected contact angle boundary
conditions for the two-phase equation, see
\cite[Eq. 4.4, Eq. 4.5]{2006_QianWangShen_Variational_MovingContactLine__BoundaryConditons}.
Here $\gamma$ denotes the fluid-solid interfacial free energy, see 
\cite[Sec. 4]{2006_QianWangShen_Variational_MovingContactLine__BoundaryConditons},
$l(\varphi)$ is a slip coefficient for the generalized Navier boundary condition applied to the
tangential part of the velocity 
$v_{tan} := v - (v\cdot \nu_\Omega)\nu_\Omega$,
while $L(\varphi)\nabla \varphi \times \nu_\Omega$ is the uncompensated Young stress and $L$ is the chemical potential at
the solid surface. 
The static contact angle is denoted by $\theta_s$ and $r \geq 0$ is a phenological parameter allowing for nonequilibrium at the contact line.
For $r\equiv 0$ \cref{eq:M:7_CH_BC} reduces to $\sigma\epsilon \nabla\varphi\cdot\nu_\Omega = - \gamma^\prime(\varphi)$, 
which means, that a static contact angle at the interface is assumed.
Furthermore, for $\gamma^\prime(\varphi) \equiv 0$ (or rather $\theta_s \equiv 90\degree$, see \Cref{rm:M:ContactEnergy}), \cref{eq:M:7_CH_BC} further simplifies to 
$\nabla\varphi\cdot\nu_\Omega = 0$, which is a no-flux condition for $\varphi$ at the solid surface.
The no-slip condition for $v$ is obtained from \cref{eq:M:6_NS_BC} by $L\equiv 0$ and $l \rightarrow \infty$ (or rather the slip length $l_s \equiv 0$, see \Cref{rm:bc_parameters}).

Concerning the existence of solutions to \cref{eq:M:1_NS1,eq:M:2_NS2,eq:M:3_CH1,eq:M:4_CH2,eq:M:5_NS_BC_1,eq:M:8_mu_neumann} 
together with no-slip for $v$ and a homogeneous Neumann (or no-flux) boundary condition for $\varphi$ as well as with different assumptions on $b$ and $W$, we refer
to \cite{AbelsDepnerGarcke_CHNS_AGG_exSol,
AbelsDepnerGarcke_CHNS_AGG_exSol_degMob,
AbelsBreit_weakSolution_nonNewtonian_DifferentDensities,
Gruen_convergence_stable_scheme_CHNS_AGG}. 
For the boundary conditions
considered here we are not aware of such results, 
but refer to 
\cite{2016-GalGrasselliMiranville-CHNSMCL-EqualDensityExSol}
for the Cahn--Hilliard Navier--Stokes system with equal densities, 
to
\cite{2017_ColliGilardiSprekels_CH_with_dynamicBoundary}
for analytical results for the Cahn--Hilliard system with dynamic boundary
conditions, 
and to \cite{GruenMetzger__CHNS_decoupled} for a Cahn--Hilliard~
Navier--Stokes model with dynamical contact angle condition, but no-slip
condition for the Navier--Stokes equation.
Concerning sharp interface limits, we refer to \cite{AbelsGarckeGruen_CHNSmodell} for the bulk model 
with homogeneous boundary conditions.
Sharp interface analysis for the  model with equal densities including contact line dynamics is available in 
\cite{2018-XuDiHu-SharpInterfaceLimit-NavierSlipBoundary}.

For \cref{eq:M:1_NS1,eq:M:2_NS2,eq:M:3_CH1,eq:M:4_CH2,eq:M:5_NS_BC_1,eq:M:8_mu_neumann} together with no-slip for $v$ and no-flux for $\varphi$, several 
thermodynamically consistent discretization schemes were proposed in the last years.
Here, we refer to \cite{Tierra_Splitting_CHNS,
Gruen_Klingbeil_CHNS_AGG_numeric,GarckeHinzeKahle_CHNS_AGG_linearStableTimeDisc}.
Especially in \cite{GarckeHinzeKahle_CHNS_AGG_linearStableTimeDisc} the influence
of spatial adaptivity on the fully discrete energy law is discussed.
We further refer to \cite{Aland__time_integration_for_diffuse_interface}, where the 
benefit of using fully coupled schemes is shown numerically, and to \cite{GonzalesTierra_linearSchemes_CH}
for an extensive discussion of several discretization schemes for the bulk energy potential $W^{poly}$.
For the full model \cref{eq:M:1_NS1}--\cref{eq:M:8_mu_neumann} thermodynamically consistent
schemes are for example proposed in \cite{AlandChen__MovingContactLine} for the case of
constant density, and in \cite{YuYang_MovingContactLine_diffDensities} for the general case.
The case with no-slip boundary condition for $v$ and dynamically advected boundary condition for $\varphi$ is
numerically and analytically considered in \cite{GruenMetzger__CHNS_decoupled}.

\begin{remark}[Bulk energy potentials]
\label{rm:M:freeEnergies}
Throughout this work we consider polynomially bounded potentials for $W$.
To state the precise assumptions we split $W = W_+ + W_-$ with $W_+$ denoting
the convex part of $W$ and $W_-$ denoting the concave part. We assume that
$W:\mathbb R \to \mathbb R$ is continuously differentiable and that
$W$ and its derivatives $W_+^\prime$ and $W_-^\prime$ are polynomially bounded,
i.e., there exists $C>0 $ such that 
\begin{align*}
|W(\varphi)| \leq C(1+|\varphi|^4),\quad
|W_+^\prime(\varphi)| \leq C(1+|\varphi|^3), \quad
|W_-^\prime(\varphi)| \leq C(1+|\varphi|^3).
\end{align*}
Note, that these bounds on the polynomial degree might be relaxed, see
\cite[(A3)]{GarckeHinzeKahle_CHNS_AGG_linearStableTimeDisc}, and that these
assumptions are used to show the existence of discrete solutions.

These assumptions are for 
example fulfilled by the commonly used polynomial potential 
\begin{align*}
W^{poly}(\varphi) := \frac{1}{4}(1-\varphi^2)^2,
\quad
W^{poly_2}(\varphi) := 
\begin{cases}
\frac{1}{4}(1-\varphi^2)^2 & \mbox{if } |\varphi|\leq 1,\\
(|\varphi|-1)^2 & \mbox{else,}
\end{cases} 
\end{align*}
where $W^{poly_2}$ is a modification of $W^{poly}$ that guarantees an $L^\infty$ bound on 
$\varphi$, see \cite{1995_CaffarelliMulder_LinftyForCahnHilliard}.

Another potential that fulfills the assumptions is 
\begin{align*}
 W^{\hp}(\varphi) := \frac{1}{2}\left( 1 - (\xi\varphi)^2 
 + \hp\lambda(\xi \varphi)^2 \right) + \theta,
\end{align*}
where $\lambda(x) := \max(0,x-1) + \min(0,x+1) $, $\theta := \frac{1}{2(\hp-1)}$ 
and $\xi := \frac{\hp}{\hp-1}$ 
are chosen such that $W(\pm1) \equiv 0$ are the
two minima of $W^{\hp}$. Here 
$\hp\gg1$ is a penalization parameter.
It appears as Moreau--Yosida relaxation of the double obstacle 
potential $W^{\infty}$, see  \cite{BloweyElliott_I,HintermuellerHinzeTber}.
In a synthetic rising bubble benchmark, \cite{Hysing_Turek_quantitative_benchmark_computations_of_two_dimensional_bubble_dynamics}, 
our results with this potential are typically closer to the
results from sharp interface methods than with the potential $W^{poly}$, see
\cite[Tab. 1]{GarckeHinzeKahle_CHNS_AGG_linearStableTimeDisc}.

In the following, whenever we use the letter $W$, we mean any of the three
mentioned bulk energy potentials.

In preparation of later results, we state the splittings of the potentials  $W$ into $W(\varphi)= W_+(\varphi) + W_-(\varphi)$.
 These are
 \begin{align*}
 W^{poly}_+(\varphi) &= \frac{1}{4}\varphi^4  - \frac{1}{4},
  & W^{poly}_-(\varphi) &= \frac{1}{2}(1- \varphi^2),\\
W^{poly_2}_+(\varphi) &= 
  \begin{cases}
  \frac{1}{4}\varphi^4 - \frac{1}{4} & \mbox{if }  |\varphi| \leq 1,\\
	(|\varphi|-1)^2-\frac{1}{2}(1-\varphi^2) & \mbox{if } |\varphi| > 1,
  \end{cases}
  &
    W^{poly_2}_-(\varphi) &= \frac{1}{2}(1-\varphi^2),\\
W^{\hp}_+(\varphi) &= \frac{\hp}{2}\lambda(\xi \varphi)^2 + \theta,
  & W^{\hp}_-(\varphi) &= \frac{1}{2}(1-(\xi\varphi)^2).
 \end{align*}
 These splittings are not unique, and we refer for example to 
 \cite{2014-WuZwietenZee-StabilizedSecondOrderConvecSplittingCHmodels} for an alternative splitting of $W^{poly_2}$.
 We further refer to \cite{GonzalesTierra_linearSchemes_CH} for a discussion on the dissipation that is introduced by the convex-concave splitting
 and also for an elaborated discussion on the dissipation that in general is introduced by splitting $W$.
 In our numerical tests, splittings that have a quadratic convex part and thus give linear systems, 
 typically lead to broader interfaces during the simulation and require smaller
 time steps to prevent this effect. Thus it is favorable to use non-linear systems as obtained by
 the proposed splittings above. 
 
To define the scaled surface tension $\sigma$ we introduce the constant $c_W$ as 
$c_W^{-1} = \int_{-\infty}^\infty 2W(\Phi_0(z))\,dz  
= \int_{-\infty}^\infty (\partial_z\Phi_0(z))^2\,dz$, 
where $\Phi_0$ denotes the first order approximation of $\varphi$ depending on $W$.
It satisfies
$\Phi_0(z)_{zz} = W^\prime(\Phi_0(z))$,
 see \cite[Sec. 4.3.3]{AbelsGarckeGruen_CHNSmodell}.
Then $\sigma = c_W \sigma_{12}$, where $\sigma_{12}$ denotes the physical value
of the surface tension between phase $1$ and phase $2$.
As the dynamics of the diffuse model depend on the particular form of $W$, this scaling is necessary to
guarantee that the same sharp interface dynamic is approximated independently of $W$.
Using $W^{poly}$ and
$W^{poly_2}$ it holds $\Phi_0(z) = \tanh(z/\sqrt{2})$ and $c_W = \frac{3}{2\sqrt 2}$.
For $W^\hp$ one obtains by elementary calculation
\begin{align*}
  \Phi_0(z) = 
  \begin{cases}
    -\Phi_0(-z) & \mbox{if } z<0,\\
    \sqrt{\xi}^{-1}\sin(\xi z) & \mbox{if } 0\leq z \leq z_0 := \xi^{-1}\arctan(\sqrt{s-1}),\\
    1-s^{-1}\exp(-\xi\sqrt{s-1}(z-z_0)) & \mbox{if } z > z_0, 
  \end{cases}
\end{align*}
and 
\begin{align*}
c_W^{-1} = 
(1-\hp^{-2})\arctan(\sqrt{\hp-1})
+\hp^{-2}(\hp+2)\sqrt{\hp-1}.
\end{align*}
For $\hp\to\infty$ we
recover the well-known scaling $c_W = \frac{2}{\pi}$ for the double-obstacle potential. 
\end{remark}

\begin{remark}[Contact line energy]
  \label{rm:M:ContactEnergy}
  The basic formula to derive the contact line energy is given by Young's law, namely
  \begin{align*} 
    \sigma_{s1}-\sigma_{s2} = \sigma_{12}\cos\theta_s.
  \end{align*}
  Here $\sigma_{s1}$ and $\sigma_{s2}$ denote the physical surface tensions between phase 1 ($\varphi=-1$) and the solid
  ($\sigma_{s1}$) and phase 2 ($\varphi=1$) and the solid ($\sigma_{s2}$).
  Further $\sigma_{12}$ 
  denotes the surface tension between phase 1 and phase 2 and $\theta_s$ denotes the
  static equilibrium contact angle between the solid and the interface and is
  measured in phase 2.
  
 We use the ansatz 
  \begin{align*}
    \gamma(\varphi) := 
    \frac{\sigma_{s1}+\sigma_{s2}}{2} 
    -\sigma_{12}\cos\theta_s\vartheta(\varphi)
  \end{align*} 
  and choose $\vartheta(\varphi)$ to fulfill 
  \begin{align*}
    \gamma(-1) = \sigma_{s1},
    \quad \gamma(0) = \frac{\sigma_{s1}+\sigma_{s2}}{2},
    \quad \gamma(1) = \sigma_{s2},
     \quad \gamma^\prime(\pm1) = 0.
  \end{align*}
  In particular it holds, that $\vartheta(-1) = -\frac{1}{2}$ and $\vartheta(1)=\frac{1}{2}$. 
  Here, the unscaled value of the surface tension appears as
  can be shown by matched asymptotic expansions, see 
  \cite[Sec. 4.3.4]{AbelsGarckeGruen_CHNSmodell}.
   
   Common choices for $\vartheta$ contain the sine function $\vartheta^{\sin}(\varphi):=
  \frac{1}{2}\sin(\frac{\pi}{2}\varphi)$, for example proposed in 
  \cite[Sec.  4]{2006_QianWangShen_Variational_MovingContactLine__BoundaryConditons},
 or a cubic polynomial $\vartheta^{poly}(\varphi) = \frac{1}{4}(3\varphi-\varphi^3)$, 
 for example proposed in \cite{2018-XuDiHu-SharpInterfaceLimit-NavierSlipBoundary}.
 An alternative is given in
 \cite{2007_DingSpelt_WettingCondition_inDiffuseInterface_ContactLineMotion}.
  Here the assumption of equipartition of energy, i.e., $\frac{\epsilon}{2}|\nabla \varphi|^2 \approx \frac{1}{\epsilon}W(\varphi)$, is used
  to derive $(\vartheta^{W})^\prime(\varphi) = c_W\sqrt{2W(\varphi)}$.
  Finally, we state a contact line energy, that is the sum of a convex and a concave function
  namely  
  $\vartheta^{cc}(\varphi) = \vartheta^{cc}_+(\varphi) + \vartheta^{cc}_-(\varphi)$ with
\begin{align*}
  \vartheta^{cc}_+(\varphi) &= 
  \begin{cases}
  -\frac{1}{2} & \mbox{if } \varphi\leq-1,\\
  \frac{1}{2}(\varphi+1)^2 -\frac{1}{2} & \mbox{if } \varphi \in (-1,0),\\
   \varphi& \mbox{if } \varphi \geq 0,
  \end{cases} 
&
\vartheta^{cc}_-(\varphi) &= 
  \begin{cases}
    0 & \mbox{if } \varphi \leq 0,\\
    -\frac{1}{2}\varphi^2 & \mbox{if } \varphi \in (0,1),\\
   \frac{1}{2}-\varphi & \mbox{if } \varphi \geq 1.
   \end{cases} 
  \end{align*}
  Here $\vartheta^{cc}_+$ is convex and $\vartheta^{cc}_-$ is concave and
  $\vartheta^{cc} \in C^{1,1}(\mathbb R)$ with $\vartheta^{\prime\prime} \in
  L^{\infty}(\mathbb R)$.

  Note, that for any $\vartheta$ that has a bounded second derivative, we can
  define a convex-concave splitting via 
  \begin{align*}
    \vartheta_+(\varphi) &= \vartheta(\varphi) + \frac{1}{2}\max_{\phi\in\mathbb R}(\vartheta^{\prime\prime}(\phi))\varphi^2,
&
\vartheta_-(\varphi) &= - \frac{1}{2}\max_{\phi\in\mathbb R}(\vartheta^{\prime\prime}(\phi))\varphi^2,
  \end{align*}
  compare 
  \cite{2018_BackofenVoigt_ConvexitySplittingPhaseField}.
  This is very similar to the stabilization approach, proposed for example in \cite{2015-ShenYangYu-EnergyStableSchemesForCHMCL-Stabilization}, 
  that essentially resembles one of Eyre's linear schemes \cite{GonzalesTierra_linearSchemes_CH}.
In the following we always assume a convex-concave splitting of $\gamma$.
  This approach can also be used for the potential $W$.

 \end{remark}

 \begin{remark}
 	To the best of our knowledge, there is no consent yet which combinations of bulk energy potential and contact line energy are most appropriate from both a physical and numerical point of view. 
 	From an analytical point of view, all combinations are reasonable
 	that lead to the correct sharp interface limit, 
 	see \cite{2018-XuDiHu-SharpInterfaceLimit-NavierSlipBoundary} for results on formal
 	sharp interface asymptotics. 
	Here, the authors use the  
 	combination of $W^{poly}$ and $\vartheta^{poly}$. 
 	However, this topic is subject to future work.
 	Further note, that using the notation from 
 	\cite{2018-XuDiHu-SharpInterfaceLimit-NavierSlipBoundary} 
 	we are in the setting $L_d = \mathcal O(\epsilon)$, and $V_s = \mathcal O(1)$.
 \end{remark}

\section{The weak formulation}
\label{sec:F}
We next derive the weak formulation that is the basis for our numerical
scheme proposed in \cref{sec:S}.
We assume sufficient regularity of all appearing functions.
Multiplying \cref{eq:M:3_CH1} with $\frac{\partial\rho}{\partial\varphi}$ we
observe
\begin{align}
\partial_t \rho + \mbox{div}(\rho v+J) = R.
\label{eq:M:MassConservationLaw}
\end{align}
Note that if $\rho$ is a  nonlinear function $R \neq 0$ holds and thus mass
conservation can be violated as soon as a nonlinear function for $\rho$ is used
to guarantee $\rho>0$.
Note that the conservation of $\varphi$ is not affected.
 Using \cref{eq:M:MassConservationLaw} the momentum equation \cref{eq:M:1_NS1}
 can equivalently be written as
\begin{align}
\partial_t(\rho v) 
+ \mbox{div}\left(v \otimes (\rho v + J)\right)
- R\frac{v}{2}
-\mbox{div}\left(2\eta Dv\right)
+\nabla p &=- \varphi\nabla\mu  + \rho g,
\label{eq:M:1.2_NS1.2}
\end{align}
see
\cite[Eq. (1.12)]{AbelsBreit_weakSolution_nonNewtonian_DifferentDensities}. We stress that this
reformulation is independent of the actual boundary condition.

To define the weak formulation we multiply both \cref{eq:M:1_NS1} and \cref{eq:M:1.2_NS1.2} by a
solenoidal test function $\frac{1}{2}w$ that satisfies $w|_{\partial\Omega}\cdot \nu_\Omega = 0$
and sum up the equations to achieve
\begin{align*}
  \frac{1}{2}\int_\Omega(\rho \partial_t v + \partial_t(\rho v)) \cdot w\dx
  - \int_\Omega \mbox{div}\left(2\eta Dv\right) \cdot w\dx +\int_\Omega (\varphi\nabla \mu-\rho g)\cdot w\dx&
  \nonumber\\
  +\frac{1}{2} \int_\Omega ((\rho v + J) \cdot \nabla) v \cdot w\dx
  +\frac{1}{2} \int_\Omega \mbox{div}\left(v \otimes (\rho v + J)\right)  \cdot w\dx 
  = 0. &
\end{align*}
Using integration by parts  together with the boundary
conditions $v\cdot \nu_\Omega = 0$ and $\nabla \mu \cdot \nu_\Omega = 0$ we observe
\begin{align*}
    &\frac{1}{2} \int_\Omega ((\rho v + J)\cdot  \nabla) v \cdot w\dx
  +\frac{1}{2} \int_\Omega \mbox{div}\left(v \otimes (\rho v + J)\right)\cdot w\dx\\
  =& \frac{1}{2}\int_\Omega ((\rho v + J)\cdot \nabla) v \cdot  w - ((\rho v + J)\nabla) w \cdot  v \dx\\
  =:& a(\rho v + J,v,w).
\end{align*}
Note that $a(\cdot,v,v) = 0$ holds.
Using integration by parts for the viscous stress we observe
\begin{align*}
   - \int_\Omega \mbox{div}\left(2\eta Dv\right)\cdot w\dx =&
    \int_\Omega 2\eta Dv : Dw\dx 
   -\int_{\partial\Omega} 2\eta Dv \nu_\Omega\cdot  w \ds,\\
   =&\int_\Omega 2\eta Dv : Dw\dx 
   +\int_{\partial\Omega} (l(\varphi)v_{tan} +  rB\nabla
   \varphi) \cdot w \ds   
\end{align*}
where $Dv:Dw := \sum_{ij=1}^n (Dv)_{ij}(Dw)_{ij}$
and we use the boundary conditions \cref{eq:M:6_NS_BC} and
\cref{eq:M:7_CH_BC}.

The weak form of \cref{eq:M:3_CH1}--\cref{eq:M:4_CH2} is derived by the standard procedure.
Summarizing the equations, we obtain the following weak form of \cref{eq:M:1_NS1}--\cref{eq:M:8_mu_neumann}:
\begin{definition}[The weak formulation]
Find sufficiently smooth $v,\mu,\varphi$, with $v$ solenoidal, $v\cdot \nu_\Omega = 0$, such that for all
$w$, $\psi$, $\phi$, with $w$ solenoidal, the following equations are satisfied:
\begin{align}
   \frac{1}{2}\int_\Omega(\rho \partial_t v + \partial_t(\rho v))\cdot  w\dx
  +a(\rho v + J,v,w)
  +\int_\Omega 2\eta Dv : Dw\dx&\nonumber\\
  +\int_{\partial\Omega} (l(\varphi)v_{tan} +  r B(\varphi_t,\varphi,v)\nabla \varphi) \cdot w \ds
    -\int_\Omega(-\varphi\nabla \mu +\rho g )\cdot w\dx &= 0,
  \label{eq:M:weak_1}
  \\
\int_\Omega \varphi_t \psi\dx 
  - \int_\Omega \varphi v\cdot \nabla \psi\dx 
  +\int_\Omega b\nabla \mu\cdot\nabla \psi\dx &= 0,
  \label{eq:M:weak_2}\\
\int_\Omega
   \sigma\epsilon\nabla \varphi\cdot\nabla \phi
   + \frac{\sigma}{\epsilon}W^\prime(\varphi)\phi \dx
   - \int_\Omega\mu \phi \dx 
  &\nonumber\\
  +\int_{\partial\Omega} \left( r B(\varphi_t,\varphi,v) + \gamma^\prime(\varphi)\right)
  \phi \ds   &= 0.
    \label{eq:M:weak_3}
\end{align} 
\end{definition}
The weak form \cref{eq:M:weak_1}--\cref{eq:M:weak_3} allows us to derive the following energy
identity.

\begin{theorem}[The formal energy identity]
Assume there exists a sufficiently smooth solution to
\cref{eq:M:weak_1}--\cref{eq:M:weak_3}. Then the following energy identity
holds
\begin{equation}
  \begin{aligned}
       \frac{d}{dt}\left(\int_\Omega \frac{1}{2}\rho |v|^2\dx
  + \sigma\int_\Omega\frac{\epsilon}{2}|\nabla \varphi|^2 +
  \frac{1}{\epsilon}W(\varphi)\dx
   + \int_{\partial\Omega}\gamma\ds
  \right)\\
  + \int_\Omega 2\eta |Dv|^2\dx
  	+ \int_\Omega b|\nabla \mu|^2\dx
  	\\
+\int_{\partial\Omega}l(\varphi)|v_{tan}|^2\ds
  + r \int_{\partial\Omega} |B(\varphi_t,\varphi,v)|^2\ds = \int_\Omega \rho g \cdot v\dx.
      \label{eq:M:EnergyIdent} 
  \end{aligned}
\end{equation}
Note that the energy in the system can only increase by the gravitational acceleration.
\end{theorem}
\begin{proof}
Use $w \equiv v$, $\Psi \equiv  \mu$, and $\Phi \equiv  \partial_t \varphi$ as test
functions in \cref{eq:M:weak_1}--\cref{eq:M:weak_3} and sum up the resulting equations. 
\end{proof}

\section{The numerical schemes}
\label{sec:S}For a practical implementation in a finite element scheme we introduce a time
grid $0 = t_0 < t_1 < \ldots<t_{m-1} < t_m< \ldots <t_M = T$ on $I = [0,T]$.
For the sake of notational simplicity let the time grid be equidistant with step size
$\tau>0$.
We further introduce 
a triangulation $\mathcal T_h$ of $\overline \Omega$
into cells $T_i$, such that $\mathcal T_h = \bigcup_{i=1}^{N}T_i$ covers
$\overline \Omega$ exactly. 

On $\mathcal T_h$ we introduce the finite element spaces
\begin{align*}
V_1 &:= \{v \in C(\overline\Omega) \,|\, v|_{T_i} \in \mathcal P_1\},\\
V_2 &:= \{v \in C(\overline\Omega)^d \,|\, v|_{T_i} \in (\mathcal P_2)^2,\,
v \cdot \nu_\Omega = 0\},
\end{align*}
where $\mathcal P_k$ denotes the space of polynomials of order up to $k$.
We use $V_1$ to define discrete approximations $\varphi_h$, $\mu_h$, and $p_h$ of
the corresponding continuous variables, and $V_2$ to define the discrete
approximation $v_h$ of $v$. 
This means that we use standard Taylor--Hood elements for the
Navier--Stokes part and explicitly denote the pressure variable in the following.

The scheme reads as follows:\\
Given 
$\varphi^{m-1} \in V_1$, 
$\mu^{m-1} \in V_1$, and
$v^{m-1} \in V_2$, 
find 
$\varphi^m_h \in V_1$, 
$\mu^m_h \in V_1$,
$p^m_h \in V_1$ and
$v^m_h \in V_2$,
such that for all 
$w \in V_2$,
$q \in V_1$,
$\Phi \in V_1$, and
$\Psi \in V_1$
the following equations hold
\begin{align}
  \frac{1}{\tau}\left( \frac{\rho^m+\rho^{m-1}}{2} v^m_h -\rho^{m-1}v^{m-1},w\right)\nonumber\\
  + a(\rho^{m-1}v^{m-1} + J^{m-1},v^m_h,w)
  + (2\eta^{m-1}Dv^m_h,Dw) - (\mbox{div} w,p^m_h)\nonumber\\
   + (l(\varphi^{m-1})v^m_{h,tan} +  r B^m_h\nabla \varphi^{m-1},w)_{\partial\Omega}\nonumber\\
		+( \varphi^{m-1}\nabla \mu^m_h,w) -(g\rho^{m-1},w) &= 0,
    \label{eq:S:1_NS_1}\\
-(\mbox{div} v^m_h,q) &= 0, \label{eq:S:2_NS_2}\\
\frac{1}{\tau}(\varphi_h^{m} - \varphi^{m-1},\Psi)
   -(\varphi^{m-1}v^m_h,\nabla \Psi)
		+(b\nabla \mu^m_h,\nabla \Psi) &= 0 \label{eq:S:3_CH1},\\
\sigma \epsilon(\nabla \varphi^m_h,\nabla \Phi)
   +\frac{\sigma}{\epsilon}(W_+^\prime(\varphi^m_h) + W_-^\prime(\varphi^{m-1}),\Phi)  
   - (\mu^m_h,\Phi) \nonumber\\
   +\left( r B^m_h,
    \Phi\right)_{\partial\Omega} 
    +\left( \gamma^\prime_+(\varphi^{m}_h) + \gamma^\prime_-(\varphi^{m-1}),
    \Phi\right)_{\partial\Omega} 
    &= 0,
    \label{eq:S:4_CH2}
\end{align}
with 
$J^{m-1} := -b\frac{\partial \rho}{\partial\varphi}(\varphi^{m-1})\nabla \mu^{m-1}$,
$B^m_h:= \left( 
\frac{\varphi^m_h-\varphi^{m-1}}{\tau} + v^m_h\cdot \nabla \varphi^{m-1}
\right)$,
$\rho^{m-1} := \rho(\varphi^{m-1})$, and $\eta^{m-1} := \eta(\varphi^{m-1})$.

\bigskip

Using Brouwer's fixed-point theorem 
one can show the existence of at least one solution 
following \cite[Thm. 2]{GarckeHinzeKahle_CHNS_AGG_linearStableTimeDisc}. 
The uniqueness stays unclear due to the nonlinearity
$\rho^mv^m_h$ in \cref{eq:S:1_NS_1}.
The scheme fulfills a fully discrete variant of the formal energy identity \cref{eq:M:EnergyIdent}.

\begin{theorem}[The fully discrete energy inequality]
\label{thm:M:enerInequDisc}
Let $\varphi^m_h \in V^m_1$,
$\mu_h^m \in V^m_1$, and
$v_h^m \in V_2^m$ denote a solution to 
\cref{eq:S:1_NS_1}--\cref{eq:S:4_CH2}.
Then the following energy inequality holds
\begin{equation*}
  \begin{aligned}
    \frac{1}{\tau}\left(\frac{1}{2}\int_\Omega \rho^m|v^m_h|^2 
    +\sigma\int_\Omega\frac{\epsilon}{2}|\nabla \varphi^m_h|^2
    + \frac{1}{\epsilon}W(\varphi^m_h)\dx
    + \int_{\partial\Omega} \gamma(\varphi^m_h)\ds
    \right)\\
+\int_\Omega 2\eta^{m-1}|Dv^m_h|^2 \dx 
    + b \int_\Omega |\nabla \mu^m_h|^2\dx
     +\int_{\partial\Omega}l(\varphi^{m-1})|v^m_{h,tan}|^2\ds
    + r \int_{\partial\Omega} |B^m_h|^2\ds    \\
+ \frac{1}{\tau}\left( 
     \frac{1}{2}\int_\Omega \rho^{m-1}|v^m_h-v^{m-1}|^2\dx
     + \frac{\sigma\epsilon}{2}\int_\Omega |\nabla \varphi^m_h-\nabla \varphi^{m-1}|^2\dx
     \right)\\
\leq \frac{1}{\tau}\left(\frac{1}{2}\int_\Omega \rho^{m-1}|v^{m-1}|^2 
    +\sigma\int_\Omega\frac{\epsilon}{2}|\nabla \varphi^{m-1}|^2
    + \frac{1}{\epsilon}W(\varphi^{m-1})\dx
    + \int_{\partial\Omega} \gamma(\varphi^{m-1})\ds \right)\\
    +\int_\Omega \rho^{m-1} g\cdot v^m_h\dx. 
  \end{aligned}
\end{equation*}
\end{theorem}
\begin{proof}
We use $w \equiv v^m_h$,
$q = p^m_h$, 
$\Psi \equiv  \mu^m_h$ 
and $\Phi \equiv \frac{\varphi^m_h-\varphi^{m-1}}{\tau}$ 
as test functions in \cref{eq:S:1_NS_1}--\cref{eq:S:4_CH2}
and sum up to obtain
\begin{align*}
  \frac{1}{\tau}\left(\frac{1}{2}\int_\Omega \rho^m|v^m_h|^2 -
  \frac{1}{2}\int_\Omega \rho^{m-1}|v^{m-1}|^2+
  \frac{1}{2}\int_\Omega \rho^{m-1}|v^{m}_h - v^{m-1}|^2 \dx \right)\\
  +\int_\Omega 2\eta^{m-1}|Dv^m_h|^2 \dx 
  -\int_\Omega \rho^{m-1} g\cdot v^m_h\dx\\
  +\int_{\partial\Omega}l(\varphi^{m-1})v^m_{h,tan}\cdot v^m_h\ds
  +r\int_{\partial\Omega}B^m_h\nabla \varphi^{m-1}\cdot v^m_h\ds\\
+ b \int_\Omega |\nabla \mu^m_h|^2\dx\\
  +\frac{\sigma\epsilon}{2\tau}\left( 
  \int_\Omega |\nabla \varphi^m_h|^2
   - |\nabla \varphi^{m-1}|^2
  +|\nabla \varphi^m_h-\nabla \varphi^{m-1}|^2\dx\right)\\
  +\frac{\sigma}{\epsilon}\int_\Omega (W^\prime_+(\varphi^{m}_h)+W^\prime_-(\varphi^{m-1}))
  \frac{\varphi^m_h-\varphi^{m-1}}{\tau}\dx\\
  +r\int_{\partial\Omega} B^m_h\frac{\varphi^m_h-\varphi^{m-1}}{\tau}\ds
  + \int_{\partial\Omega} 
  (\gamma_+^\prime(\varphi_h^{m}) +\gamma_-^\prime(\varphi^{m-1})) 
  \frac{\varphi^m_h-\varphi^{m-1}}{\tau}\ds
  = 0.
\end{align*}
Using convexity and concavity of $W_+$ and $W_-$, and $\gamma_+$ and $\gamma_-$
it holds
\begin{align*}
  \int_\Omega(W^\prime_+(\varphi^{m}_h)+W^\prime_-(\varphi^{m-1}))
  \frac{\varphi^m_h-\varphi^{m-1}}{\tau} \dx
  \geq \frac{1}{\tau}\int_\Omega W(\varphi^m_h)-W(\varphi^{m-1})\dx,\\
\int_\Omega(\gamma^\prime_+(\varphi^{m}_h)+\gamma^\prime_-(\varphi^{m-1}))
  \frac{\varphi^m_h-\varphi^{m-1}}{\tau} \ds
  \geq \frac{1}{\tau}\int_\Omega \gamma(\varphi^m_h)-\gamma(\varphi^{m-1})\ds.
\end{align*}
Summing up and using $v\cdot \nu_\Omega = 0$, we obtain the desired result.
\end{proof}

\begin{remark}[Adaptive meshing]
In general, in diffuse interface simulations it is advantageous to use adaptive
meshes to resolve the interfacial region. Then in every time step additional
prolongation operators between subsequent meshes are required. As a consequence, in this case
the energy inequality from \Cref{thm:M:enerInequDisc} only holds with the
prolongated data  for the energy from the old time instance.
We further note that special care has to be taken for prolongating the velocity
field, as the prolongated velocity field typically is not solenoidal with respect
to the new mesh.
We refer to
\cite{GarckeHinzeKahle_CHNS_AGG_linearStableTimeDisc,
2011-BesierWollner-PressureApproximationIncompressibleFlowDynamicallySpatialMesh} for further discussion of this topic.

\end{remark}

\subsection{Variants}
\label{ssec:S:V}
Let us state variants of the above discretization scheme \cref{eq:S:1_NS_1}--\cref{eq:S:4_CH2}
for numerical comparison.
We note, that \cref{eq:S:1_NS_1}--\cref{eq:S:4_CH2} is a fully coupled and non-linear scheme.

\subsubsection{A stable decoupled scheme} 
\label{ssec:S:V:decoupled}

If $r\equiv 0$ the scheme is only coupled by the transport term  
$(\varphi^{m-1} v^m_h,\nabla \Psi) $ in \cref{eq:S:3_CH1}. The same holds
for $l \to \infty$, which results in the commonly used no-slip condition for the
Navier--Stokes equation. In the case of no-slip conditions $B$ is independent of
$v$ and thus again the only coupling is the transport term in \cref{eq:S:3_CH1}.

In both cases we can decouple the Navier--Stokes equation and the
Cahn--Hilliard equation by using an augmented velocity field in
\cref{eq:S:3_CH1}, see for example 
\cite{Minjeaud_decoupling_CHNS,Tierra_Splitting_CHNS,GruenMetzger__CHNS_decoupled,YuYang_MovingContactLine_diffDensities}.
Here we substitute $-\int_\Omega \varphi^{m-1}v^m_h \cdot\nabla \Psi\dx$ in \cref{eq:S:3_CH1} by
\begin{align}
  -\int_\Omega \varphi^{m-1}v^{m-1}\cdot \nabla \Psi\dx 
+ \tau\int_\Omega (\rho^{m-1})^{-1}|\varphi^{m-1}|^2 
\nabla \mu^m_h\cdot\nabla \Psi\dx.
\label{eq:S:V:transport}
\end{align}
The resulting scheme is decoupled; we can first solve
\cref{eq:S:3_CH1,eq:S:4_CH2} and thereafter \cref{eq:S:1_NS_1,eq:S:2_NS_2}. This scheme is also energy
stable, as the additional integral compensates terms arising from H\"older's and  Young's inequality
to balance the first integral with the numerical dissipation $\frac{1}{2}\int_\Omega
\rho^{m-1}|v^m_h-v^{m-1}|^2\dx$. 
This scheme with no-slip conditions for Navier--Stokes and $r\equiv 0$
is analyzed in \cite{GruenMetzger__CHNS_decoupled} for different treatments of
$W^\prime$.
We also refer to \cite{KayStylesWelford} for an alternative decoupling in the case of constant density.
Here the systems are decoupled by
using $v^{m-1}$ in \cref{eq:S:3_CH1}, and the energy stability is obtained by
introducing a step size restriction for the temporal discretization.

If $r >0$, we use $v^{m-1}$ in the definition of $B^m_h$ in \cref{eq:S:4_CH2} and $v^m_h$ in the
corresponding term in \cref{eq:S:1_NS_1} and  can still derive an energy inequality containing an
error of order $r \int_{\partial\Omega} (v^m_h-v^{m-1}) \cdot \nabla \varphi^{m-1}\ds$.
In \cite{GruenMetzger__CHNS_decoupled} a no-slip condition is assumed for $v$ to decouple the
boundary conditions. Then the decoupling proposed in \cref{eq:S:V:transport} is sufficient to
decouple the Navier--Stokes and the Cahn--Hilliard equation.

We note that this scheme can be applied for any bulk energy potential that admits a convex-concave
splitting.

\subsubsection{A stable  decoupled and linear  scheme}
\label{ssec:S:V:stablelinear}
Using the decoupling  proposed in \cref{ssec:S:V:decoupled}, the only nonlinearity in the scheme
arises from $W_+^\prime$. In
\cite{2015-ShenYangYu-EnergyStableSchemesForCHMCL-Stabilization,AlandChen__MovingContactLine,YuYang_MovingContactLine_diffDensities}, a stabilized linear
scheme is used and the term $W_+^\prime(\varphi^m_h) + W_-^\prime(\varphi^{m-1})$ is substituted by 
$W^\prime(\varphi^{m-1}) + S_W(\varphi^{m}_h - \varphi^{m-1})$, where $S_W$ is a suitable
stabilization parameter.
For smooth $W$ it satisfies $S_W \geq \frac{1}{2}\max_{t} |W^{\prime\prime}(t)|$. It can be derived by Taylor expansion of
$W$ at $\varphi^{m-1}$, see for example \cite{2015-ShenYangYu-EnergyStableSchemesForCHMCL-Stabilization}.
As $W^\hp$ is of class $C^{1,1}$ only, $({W^\hp})^{\prime\prime}$ jumps at $\xi^{-1}$ from $-\xi^2$ to $(\hp-1)\xi^2$.
In this case we use $S_W \geq \hp/2$.
For large values of $\hp$  we expect that this stabilization will  prevent changes in $\varphi$ and thus
might have a deep impact on the allover dynamics. 
This is investigated in \Cref{sec:N} and especially discussed in Remark~\ref{rm:relaxation}.
To linearize $\gamma$ we substitute 
$\gamma^\prime_+(\varphi^m_h) + \gamma^\prime_-(\varphi^{m-1})$ by
$\gamma^\prime(\varphi^{m-1}) + S_\gamma(\varphi^m_h - \varphi^{m-1})$ with
$S_\gamma \geq \frac{1}{2}\max_{t}|\gamma^{\prime\prime}(t)|$ and especially $S_\gamma \geq \frac{1}{2}\sigma_{12}|\cos(\theta_s)|$ in the
case of $\gamma^{cc}$. 
Here, again $S_\gamma$ is obtained by Taylor expansion of $\gamma$ at $\varphi^{m-1}$.

\begin{remark}[Further schemes]
  For further discretization schemes of the bulk energy density $W$ we refer for example to
  \cite{GruenMetzger__CHNS_decoupled,
  GonzalesTierra_linearSchemes_CH,
  2014-WuZwietenZee-StabilizedSecondOrderConvecSplittingCHmodels}. 
  Second order schemes for the Cahn--Hilliard equation 
  are for example proposed and analyzed in 
\cite{GonzalesTierra_linearSchemes_CH,
  2014-Tierra_SecondOrder_and_Adaptivity,
  2014-WuZwietenZee-StabilizedSecondOrderConvecSplittingCHmodels,
  2016-DiegelWangWise-StabilityConvergenceSecondOrderSchemeCH,
  2018-WangYu-SecondOrderCahnHilliardStabilized}. 
Recently the Invariant Energy Quadratization approach for $W\equiv W^{poly}$ was proposed in \cite{2017-YangJu-InvariantEnergyQuadratization}.
  It is used in \cite{2018-YangYu-SecondOrderCHMCL_IEQ} for the Cahn--Hilliard moving contact line model
  together with a Crank--Nicolson and a BDF2 scheme in time.
However, typically for these schemes  
  either higher regularity than $W^s$ provides is required  for $W$,
  or the particular $W^{poly}$ is assumed and necessary.
  Moreover, unconditional energy stability is typically not proven yet.
  \end{remark}

\begin{remark}[Energy Consistency]
  \label{rm:S:energyStable}
  Considering the energy inequality from \Cref{thm:M:enerInequDisc}, the terms in the first line correspond
  to the discrete energy of the system, while the second line corresponds to the energy dissipation
  of the system, and the third line corresponds to numerical dissipation of the scheme. 
Based on this we can define four different values to define the energies in our system.
These are the energy $E^m$ at time instance $m$, the physical dissipation $\Delta^m_p$ at time instance $m$,
the energy $E_g^m$ introduced from gravity at time instance $m$,
and the numerical dissipation $\Delta^m_n$ at time instance $m$. They are defined by
\begin{align}
  E^m &:= \frac{1}{2}\int_\Omega \rho^m|v^m_h|^2\dx 
    +\sigma\int_\Omega\frac{\epsilon}{2}|\nabla \varphi^m_h|^2
    + \frac{1}{\epsilon}W(\varphi^m_h)\dx
    + \int_{\partial\Omega} \gamma(\varphi^m_h)\ds,\\
\Delta^m_p &:= \tau\int_\Omega 2\eta^{m-1}|Dv^m_h|^2 \dx 
    + \tau \int_\Omega b |\nabla \mu^m_h|^2\dx \nonumber\\
&\phantom{:= }+\tau\int_{\partial\Omega}l(\varphi^{m-1})|v^m_{h,tan}|^2\ds
    + \tau \int_{\partial\Omega} r |B^m_h|^2\ds,\label{eq:S:phyDiss} \\
E_g^m & :=\tau\int_\Omega \rho^{m-1}g\cdot v_h^m\dx,\\
\Delta^m_n &:= E^{m-1} + E_g^m - E^m - \Delta^m_p.  
\end{align}
We call a scheme thermodynamically consistent if \Cref{thm:M:enerInequDisc} is fulfilled
  without the explicit form of the numerical dissipation, thus if
\begin{equation}
 	E^m + \Delta^m_p \leq E^{m-1} + E_g^m
 	\label{eq:S:EnergyInequ}
 	\end{equation}
holds, i.e., $\Delta_n^m \geq 0$.
We investigate this energy inequality numerically in \Cref{sec:N}.  
\end{remark}

  \section{Numerics}
\label{sec:N}

In this section we numerically investigate the three schemes under consideration.
In \Cref{ssec:N:rising_bubble} we briefly give results from the well-known second benchmark in
\cite{Hysing_Turek_quantitative_benchmark_computations_of_two_dimensional_bubble_dynamics},
where no contact line motion is included, to estimate the difference of the schemes 
in the bulk.
In \Cref{ssec:N:sliding_droplet} we thereafter investigate
the behavior of the contact line for a gravity-driven droplet sliding on an inclined surface in a two-dimensional setting.

We implement the schemes in Python3 using FEniCS 2018.1.0 \cite{fenics1, fenics_book}.
For the solution of the arising nonlinear and linear systems and subsystems the software suite 
PETSc 3.8.4~\cite{petsc_webpage, petsc-user-ref, petsc-efficient} 
together with the direct linear solver MUMPS 5.1.1 \cite{mumps_1, mumps_2} are utilized.
Note, that we do not apply any preconditioning or subiterations except for the Newton iterations.

\subsection{Rising Bubble}
\label{ssec:N:rising_bubble}
At first, we discuss the accuracy of the
proposed schemes without moving contact lines.
Later on, this allows for an evaluation of the influence of the schemes on the moving contact line.
We employ the quantitative benchmark case proposed 
in~\cite{Hysing_Turek_quantitative_benchmark_computations_of_two_dimensional_bubble_dynamics}.
In~\cite{Aland_Voigt_bubble_benchmark} it is found, 
that three different diffuse interface approximations together 
with the polynomial potential $W^{poly}$ agree well with the sharp interface 
results from~\cite{Hysing_Turek_quantitative_benchmark_computations_of_two_dimensional_bubble_dynamics}.
In~\cite{GarckeHinzeKahle_CHNS_AGG_linearStableTimeDisc} the benchmark 
is used to compare to a phase field model with a relaxed double obstacle
potential.

\subsubsection{Setup}
\Cref{tab:rb_setup} lists the properties of our simulations, 
which correspond to the second benchmark case in~\cite{Hysing_Turek_quantitative_benchmark_computations_of_two_dimensional_bubble_dynamics}. 
For details on the setup we refer to the references above.
Note, that $\sigma_{12}$ denotes the physical surface tension, yielding $\sigma\approx1.24$ for $W^{\hp=100}$, $\sigma\approx1.22$ for $W^{\hp=10}$ and $\sigma\approx 2.07$ for $W^{poly_2}$.
Following~\cite{Hysing_Turek_quantitative_benchmark_computations_of_two_dimensional_bubble_dynamics}, 
we introduce a characteristic length scale $L=2r_0$, where $r_0$  equals the initial radius of the bubble, 
and a characteristic velocity scale $U=\sqrt{Lg}$. 
To classify our simulations we indicate in \Cref{tab:rb_setup} 
the dimensionless numbers Reynolds $\Rey = \frac{\rho_l U L}{\eta_l}$, 
E\"otv\"os (or Bond) $\Eo=\frac{\rho_l g L^2}{\sigma}$, 
Capillary $\Ca = \frac{\eta_l U}{\sigma}$, 
Atwood $\At = \frac{\rho_l - \rho_g}{\rho_l + \rho_g}$, 
Cahn $\Cn = \frac{\epsilon}{L}$ and P\'eclet $\Pe = \frac{L U \epsilon}{b \sigma}$, 
see~\cite{Khatavkar2006}.

We apply no-slip boundary conditions for the velocity on the top and bottom
walls, free-slip on the left and symmetry at the centerline through the bubble at $x=0.5$.
Similar to~\cite{GarckeHinzeKahle_CHNS_AGG_linearStableTimeDisc}, we set $b=10^{-3}\epsilon$ and $\epsilon=0.02$.
The time discretization step is set to different values and the final time is $t=3$.
We initialize the simulations by solving the Cahn--Hilliard equations without convection until a steady state is reached.
In total, we perform 7 distinct simulations using the three schemes from
\Cref{sec:S} with $W^{poly2}$ and $W^{\hp}$ with $\hp=100$, and one additional
simulation with $\hp=10$ for the fully linear and  stabilized scheme with
$W^{\hp}$, see the first three columns in \Cref{tab:rb_results}.
To get an impression of the influence of the discretization parameters, we use different values for $\tau$ and $h_{min}$, 
see columns four and five in \Cref{tab:rb_results}.

In \cite{Hysing_Turek_quantitative_benchmark_computations_of_two_dimensional_bubble_dynamics}
a set of benchmark parameters is used, that we define in the phase field setting as follows.

The center of mass is calculated using
\begin{align}
		(x_c, y_c) &= \frac{\int_\Omega (x, y)\frac{1+\varphi}{2}\dx}{\int_\Omega \frac{1+\varphi}{2}\dx}\;,
		\label{eq:S:BM-y}
\end{align}
where $\frac{1+\varphi}{2} = 1$ indicates  the droplet.

We define the mean velocity in unit direction $a\in \mathbb R^2$ as
\begin{align}
  v_a = \frac{\int_\Omega v \cdot a \frac{1+\varphi}{2}\dx}{\int_\Omega \frac{1+\varphi}{2}\dx}\;.
  \label{eq:S:BM-v}
\end{align}
If $a$ denotes the unit vector in rising direction, this is called rising velocity $v_r$, while if $a$ points in sliding
direction, we call this value sliding velocity $v_s$.

Finally we define the stretching of the interface as
\begin{align}
c = \frac{c_W\int_\Omega (\frac{\epsilon}{2}|\nabla \varphi_0|^2 + \frac{1}{\epsilon}W(\varphi_0))\dx}{c_W\int_\Omega (\frac{\epsilon}{2}|\nabla \varphi|^2 + \frac{1}{\epsilon}W(\varphi))\dx}\;.
		\label{eq:S:BM-c}
\end{align}
Here the denominator denotes an approximation to the length of the interface represented by $\varphi$, and the numerator
denotes the same for the initial phase field $\varphi_0$.
If $\varphi_0$ denotes a sphere, this is equivalent to the circularity as defined in
\cite{Hysing_Turek_quantitative_benchmark_computations_of_two_dimensional_bubble_dynamics} 
as the volume of the bubble is constant over time.

\begin{remark}[Choice of $\hp$ in $W^\hp$]\label{rm:relaxation}
For the choice of the relaxation parameter $\hp$ in $W^\hp$, see \Cref{rm:M:freeEnergies}, several points must be considered.
To reduce the inter-mixing between the phases and increase the rate at which the equilibrium profile of $\varphi$ 
is reestablished after a deformation, it is desirable to exhibit a large spinodal region and subsequently a small metastable region~\cite{Donaldson2011}.
The metastable region of the bulk energy potential $W^{\hp}$ is located between $1 > |\varphi|> \xi^{-1} = 1-\frac{1}{\hp}$,
while the metastable region for $W^{poly}$ is located between $1> |\varphi| > \sqrt{3^{-1}} \approx 0.577$.
Thus already for small values of $\hp$, say $\hp = 10$, the metastable region of $W^{\hp}$ is significantly smaller
than the metastable region of $W^{poly}$.
Furthermore, referring to \cite{GarckeHinzeKahle_CHNS_AGG_linearStableTimeDisc,Kahle_Linfty_bound}, the value of $\hp$
controls the deviation of the $L^\infty$ norm of $\varphi$ from 1. 
Since $\rho$ and $\eta$ directly depend on $\varphi$ a small deviation is desirable, which is achieved by 
a large value of $\hp$.

On the other hand, the stable decoupled and linear scheme, \Cref{ssec:S:V:stablelinear}, 
includes a stabilization parameter $S_W$ which has to be chosen like $S_W>\hp/2$ for $W^\hp$.
In this case a large value of $\hp$ has a severe impact on the overall dynamics
as the stabilization can be interpreted as adding the quadratic potential $\frac{S_W}{2}\|\varphi-\varphi^{m-1}\|^2$
to $W$ for given $\varphi^{m-1}$. 
For large values of $S_W$ thus $\varphi \equiv \varphi^{m-1}$ is preferred.
To show the influence of $S_W$ in the case  $W\equiv W^\hp$ we test the linear and decoupled scheme with two values of $\hp$.
\end{remark}

\begin{table}
\begin{tabular}{cccccccccc}
\toprule
$\sigma_{12}$	&$\rho_l$	&$\rho_g$	&$\eta_l$	&$\eta_g$	&$g_y$&$\Rey$&$\Eo$&$\Ca$&$\At$\\
\cmidrule(r){1-6}\cmidrule(l){7-10}
1.96			&1000		&1			&10			&0.1		&-0.98&35&125&3.5&0.99\\
\cmidrule(r){1-6}\cmidrule(l){7-10}
$\epsilon$		&$b$		&&&&&$\Cn$		&$\Pe$\\
\cmidrule(r){1-6}\cmidrule(l){7-10}
$\num[scientific-notation = true]{2e-2}$&$\num[scientific-notation = true]{2e-5}$&&&&&0.04&178\\
\bottomrule
\end{tabular}
\caption{Parameters used in the rising bubble simulations.}
\label{tab:rb_setup}
\end{table}

\subsubsection{Results}

The resulting benchmark values are listed in \Cref{tab:rb_results}.
As it is not even clear in the sharp interface simulations whether or not topological changes develop, e.g. the separation of trailing gas filaments, we compare our results only up to time instance $t=2$, see~\cite{Aland_Voigt_bubble_benchmark} . 
Our results show that all the schemes give very similar results compared to the sharp interface solution even for the significantly larger time step $\tau=0.001$ and on a coarse mesh with $h_{min}=0.0125$. 
In general, decoupling the two systems has a very small impact on the benchmark values.
For even larger $\tau=0.008$ the coupled scheme is advantageous against the decoupled schemes.
The latter might be explained by the fact, that the decoupling adds artificial diffusion of order $\tau$ to the Cahn--Hilliard system, see \eqref{eq:S:V:transport}. Thus we expect a stronger influence of this decoupling for larger values of $\tau$. 
As expected, the stabilized linear scheme together with $W^{\hp=100}$ hinders the dynamics of the rising bubble. However, the results improve significantly with smaller $\hp$.
All schemes together with $W^{\hp=100}$ give slightly better results compared to $W^{poly_2}$ except the decoupled/linear scheme. 
However, for very small $\tau$ and $h_{min}$ the results converge towards similar values.

Concerning the computational effort the difference in using $W^{\hp=100}$ or $W^{poly_2}$ is insignificant. The decoupled/nonlinear and decoupled/linear schemes are around 1.4 respectively 2.0 times faster than the coupled scheme. In the nonlinear schemes 2-3 Newton iterations are needed per time step. Note that the performance results strongly dependent on the solver and whether sophisticated preconditioning is applied. For an efficient preconditioner for the coupled/nonlinear system we refer to~\cite{Bosch2016}.

\begin{table}
    	\csvreader[table head=\toprule 
			Bulk pot.&
			Deco./Lin.?&
			$\tau$&
			$h_{min}$&
			$y_c$&
			$v_{\max}$&
			$t_{v_{\max}}$&
			$c_{\min}$&
			$t_{c_{\min}}$
			\\\midrule,
table foot=\midrule\midrule 
			ref. diffuse&
			&
			0.004&
&
			0.8994&
			0.2503&
			0.7960&
			0.6684&
			1.9760\\
			ref. sharp&
			&
			0.0001953125& 
&
			0.9154&
			0.2502&
			0.7313&
			0.6901&
			2.0000\\\bottomrule,
head to column names, 
  		late after last line=\\,
		before reading=\centering\sisetup{table-number-alignment=center, table-format=1.3,round-mode=places,round-precision=4},
		tabular={ll
			S[scientific-notation = true,table-format=0.2e1, round-precision=2]
			S[table-format=1.4, round-precision=5]
			S
			S
			S
			S
			S[round-precision=2]
			},
		late after line=\ifnumequal{\thecsvrow}{15}{\\}{\\}
			\ifnumequal{\thecsvrow}{10}{\midrule}{}
			\ifnumequal{\thecsvrow}{13}{\midrule}{}
		]{data/risingbubble/risingbubble_hysing.csv}{}{\pot&\decoupled/\linear&\dt&\dx&\yct&\vmax&\tvmax&\cmin&\tcmin}
\caption{Benchmark values for the second benchmark proposed in \cite{Hysing_Turek_quantitative_benchmark_computations_of_two_dimensional_bubble_dynamics}.
Here $y_c$ denotes the center of mass at time $t=2$, $v_{\max}$ denotes the maximum rising velocity that appears at time $t_{v_{\max}}$, and
$c_{\min}$ denotes the minimal circularity that appears at time $t_{c_{\min}}$. See \eqref{eq:S:BM-y}--\eqref{eq:S:BM-c} or 
\cite{Hysing_Turek_quantitative_benchmark_computations_of_two_dimensional_bubble_dynamics} for the definition of these values.
As reference we choose the results from the 3rd group participating in \cite{Hysing_Turek_quantitative_benchmark_computations_of_two_dimensional_bubble_dynamics} 
(ref. sharp) and for model 3 in \cite{Aland_Voigt_bubble_benchmark} (ref. diffuse, $\epsilon=0.02$). We note that in the latter 
piece wise quadratic finite elements are used for $\varphi$ and $\mu$, which is the reason, why we do not provide a value for $h_{\min}$. 
Further,  we do not provide a value for $h_{\min}$ for the reference solution in the sharp setting
because here a different numerical approach is used, that can not directly be compared with the present situation.}
\label{tab:rb_results}
\end{table}

\subsection{Sliding Droplet}
\label{ssec:N:sliding_droplet}
To compare the influence of the numerical schemes from \Cref{sec:S} on the moving contact line, we perform
simulations of single droplets sliding down an inclined surface.
Besides the effect of gravity on the droplet movement, this test case allows to observe both an advancing and receding contact line.

\subsubsection{Setup}
In \Cref{fig:sd_setup} the initial configuration is shown and \Cref{tab:sd_setup} lists the properties of our simulations. 
The fluid properties are chosen to be similar to the first rising bubble test case 
in~\cite{Hysing_Turek_quantitative_benchmark_computations_of_two_dimensional_bubble_dynamics}.
Note, that $\sigma_{12}$ denotes the physical surface tension, yielding $\sigma\approx15.58$ for $W^{\hp=100}$, $\sigma\approx15.34$ for $W^{\hp=10}$ and $\sigma\approx25.98$ for $W^{poly_2}$.
A liquid droplet with radius $r_0=0.25$ is placed in a $0.5\times2.0$ rectangular domain at ($0$,$1.5$) on a smooth, solid surface with an initial contact angle of $90\degree$.
The inclination angle of the plate is $45\degree$.
The density of the droplet is greater than that of the surrounding fluid.
We have no-slip boundary conditions for the velocity on the left and right side and free-slip on the top side.
The conditions~\cref{eq:M:6_NS_BC,eq:M:7_CH_BC} are applied on the bottom solid surface, see \Cref{fig:sd_setup}.
The influence of the boundary conditions~\cref{eq:M:6_NS_BC,eq:M:7_CH_BC} on the sliding droplets 
are examined by varying the static contact angle $\thetaori$, the relaxation factor $r$ and the slip coefficient $l$, 
see the fifth to seventh column in \Cref{tab:sd_results}. 
We vary the contact angle from super-hydrophilic ($5\degree$) to super-hydrophobic ($150\degree$)~\cite{Law2014}.
We initialize the simulations by solving the Cahn--Hilliard equations without convection and a contact angle of $90\degree$ until a steady state is reached.

In a first step, we compare 21 distinct simulations obtained with the three schemes from \Cref{sec:S} with $W^{poly_2}$ and $W^{\hp}$ with $\hp=100$, and one additional
simulation with $\hp=10$ for the fully linear and stabilized scheme with $W^{\hp}$, see the first two columns in \Cref{tab:sd_results}. 
These simulation are performed with a relatively coarse mesh ($h_{min} = 0.0125)$ and large time step ($\tau=0.001$) to discuss the practical applicability of the solution schemes.
Afterwards, we show the thermodynamic consistency of the schemes and compare the physical and numerical dissipation rates.
To discuss the influence of the time step size on the results, we perform 14 additional simulations with $\tau$ between 0.008 and 0.00025.
Finally, we perform 8 simulations with varying interfacial thicknesses $\epsilon$ on a very fine mesh ($h_{min}=0.0002$) to briefly discuss the convergence to the sharp-interface limit.
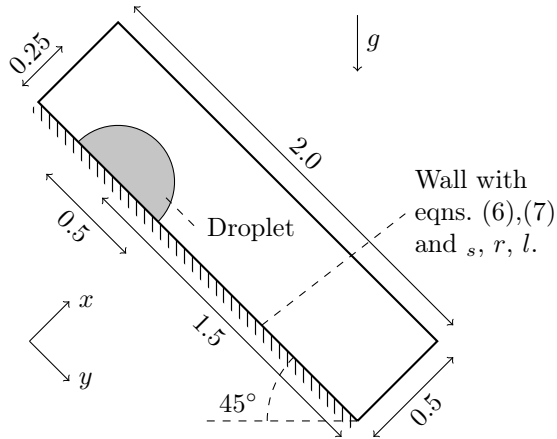
\begin{figure}
\centering
\tikzexternalexportnextfalse 
\tikzsetnextfilename{slidingdroplet_setup}
\begin{tikzpicture}[scale=3]
	\filldraw[rotate=45, fill=lightgray, fill opacity=0.9] (0, 1.75) node(r1){} arc(90:-90:0.25)node(r2){};
\draw[thick,rotate=45] (0,0) node (a){} -- (0.5,0) node (b){} -- (0.5,2) node (c){} -- (0,2) node (d){} -- (0,0);
\fill[rotate=45,pattern=vertical lines] (a) -- ++(0, 2) -- ++(-0.05,0) -- ++(0,-2) -- ++(0.05,0); 
\draw[rotate=45, <->] ($(a)+(0,-0.1)$)--node[rotate=45,below] {$0.5$} ($(b)+(0,-0.1)$) ;
	\draw[rotate=45, <->] ($(b)+(+0.1,0)$)--node[rotate=-45,above] {$2.0$} ($(c)+(+0.1,0)$) ;
\draw[rotate=45, <->] ($(a)+(-0.1,0)$)--node[rotate=-45,below] {$1.5$} ($(a)+(-0.1,1.5)$) ;
	\draw[rotate=45, <->] ($(r1)+(-0.2,0)$)--node[rotate=-45,below] {$0.5$} ($(r2)+(-0.2,0)$) ;
	\draw[rotate=45, <->] ($(d)+(0,+0.1)$)--node[rotate=45,above] {$0.25$} ($(d)+(0.25,+0.1)$) ;
\draw[dashed] (0,0)--(-0.7,0);
	\draw[dashed] (-0.4,0) node[above left]{$45\degree$} arc(180:135:0.4);
\begin{scope}[rotate around={45:(-1.1,0)}]
	\draw[->] (-1.1,0.5)--(-0.85,0.5) node[right]{$x$};
	\draw[->] (-1.1,0.5)--(-1.1,0.25) node[right]{$y$};
	\end{scope}
\draw[dashed, rotate=45] (0.1,1.3)--(0.1,1.1) node[right]{Droplet};
\draw[dashed, rotate=45] (0,0.6)--(0.8,0.5) node[right, align=left]{Wall with\\eqns.~\cref{eq:M:6_NS_BC},\cref{eq:M:7_CH_BC}\\and $\theta_s$, $r$, $l$.};
\draw[->] (0.0,1.8)-- node[right]{$g$}(0.0,1.55);
\end{tikzpicture}

 \caption{Initial configuration for the sliding droplet simulations.}
\label{fig:sd_setup}
\end{figure}

\begin{table}
\begin{tabular}{cccccccccc}
\toprule
$\sigma_{12}$	&$\rho_l$	&$\rho_g$	&$\eta_l$	&$\eta_g$	&$g$&$\Rey$&$\Eo$&$\Ca$&$\At$\\
\cmidrule(r){1-6}\cmidrule(l){7-10}
24.5			&1000		&100		&10			&1			&-0.98&35&10&0.28&0.81\\
\cmidrule(r){1-6}\cmidrule(l){7-10}
$\epsilon$		&$b$		&		&&&&$\Cn$		&$\Pe$\\
\cmidrule(r){1-6}\cmidrule(l){7-10}
$\num[scientific-notation = true]{2e-2}$&$\num[scientific-notation = true]{2e-5}$&&&&&0.04&14\\
\bottomrule
\end{tabular}
\caption{Parameters used in the sliding droplet simulations.}
\label{tab:sd_setup}
\end{table}

\begin{remark}[Choice of $r$ and $l$]\label{rm:bc_parameters}
For meaningful values of the relaxation parameter $r$ and the slip coefficient $l$, 
we write \cref{eq:M:6_NS_BC,eq:M:7_CH_BC} in non-dimensionalized form,
\begin{align}
\left[\frac{\Ca}{\Cn} 2\hat \eta D\hat v \nu_\Omega + \frac{\Ca}{\Cn} \frac{L}{l_s} - \hat L \hat \nabla \varphi\right] \times \nu_\Omega &= 0, \\
\frac{\Ca}{\Cn}\frac{r_s}{L}(\partial_{\hat t}\hat\varphi + \hat v \hat\nabla \varphi) + \hat L&=0,
\end{align}
in which $\Ca=\eta_l U/\sigma$ and $\Cn=\epsilon/L$ are the Capillary number 
respectively the Cahn number, and $L$ and $U$ are some characteristic macroscopic length scale respectively velocity.
We choose $r=r_s\eta_l$ and $l=\eta_l/l_s$ such that the dimensionless groups $\frac{\Ca}{\Cn} \frac{L}{l_s}$ 
and $\frac{\Ca}{\Cn}\frac{r_s}{L}$ are of~$\mathcal{O}(1)$, see~\cite{Sibley2013}.
\end{remark}

As benchmark values we again use the three values defined in \Cref{ssec:N:rising_bubble} with minor modifications. 
For the center of mass, we use a coordinate system that is aligned with the inclined plate, see \Cref{fig:sd_setup},
and for the now called sliding velocity, we use for $a$ the unit vector in tangential direction to the inclined plate.
The stretching of the interface is defined as before.

Additionally, we evaluate two values that are specific for the moving contact line setup.
For both the receding and advancing contact line the position of the contact points 
and a dynamic (or apparent) contact angle measured at some height above the contact points are evaluated.
The position of a contact point is defined by
\begin{align}
	y_p = y \mbox{ on } \partial\Omega \mbox{ where }\varphi=0\;,
\end{align}
and the dynamic contact angle $\thetaori_d$ is calculated by linear interpolation between $y_p$ and 
the intersection $y_{p+\Delta}\mbox{ where }\varphi=0$ and $\Delta=h_\text{min}$, 
see \Cref{fig:post_dyn_contactangle} and~\cite{Omori2017}.

\begin{figure}
\centering

\tikzsetnextfilename{dyntheta_post}

\tikzset{cross/.style={cross out, draw=black, minimum size=2*(#1-\pgflinewidth), inner sep=0pt, outer sep=0pt},
cross/.default={1pt}}

\begin{tikzpicture}[scale=3]
	\filldraw[fill=lightgray, rotate=-45, fill opacity=0.9] (0.25, 0) arc(180:0:0.25);
\draw[thick, rotate=-45] (0,0) node (a){} -- (1,0) node (b){};
\fill[pattern=vertical lines, rotate=-45] (0,0) -- (1,0) -- (1,-0.05) -- (0,-0.05) -- (0, 0); 
\draw[rotate=-45] (0.75,0) circle (0.10);
	\draw[rotate=-45] (0.85, 0) -- (1.24,0.5) node(r1){};
	\draw[rotate=-45] (0.65, 0) -- (0.26,0.5) node(r2){};
	\draw[rotate=-45] (0.75,1.0) circle (0.70);
	\clip[rotate=-45] (0.75,1.0) circle (0.70);
\draw[thick, rotate=-45] (0,0.6) -- (2,0.6);
\fill[pattern=vertical lines, rotate=-45] (0,0.6) -- (2,0.6) -- (2,0.295) -- (0,0.295) -- (0, 0.6); 
\draw[rotate=-45, thick] (0.75, 0.6) to[out=60,in=-40] node[above right](bla){$\mathclap{\varphi=0}$}(0.5, 1.8);
	\draw[rotate=-45, dashed] (0.75, 0.6) node[draw, circle, minimum size=3.0pt, solid] {} -- node[draw, circle, minimum size=3.0pt, solid](cross2){} (1.0, 1.5);
\draw[rotate=-45, <->] (1.1, 0.6)--node[ below right, rotate=-0]{$h_\text{min}$}(1.1, 1.05);
\draw[rotate=-45, dashed] (0.45, 0.6) node[rotate=0, above]{$\thetaori_d$} arc(180:75:0.3);
\begin{pgfonlayer}{background}
		\clip[rotate=-45] (0.75,1.0) circle (0.70);
			\draw[dashed, rotate=-45] (-1, 1.05) -- (2, 1.05);
		\fill[rotate=-45,lightgray, fill opacity=0.5] (0,0.6) -- (0.75, 0.6) to[out=60,in=-40] (0.5, 1.8) -- (0, 1.8) -- (0, 0.6);
	\end{pgfonlayer}
\end{tikzpicture}

 \caption{Measurement of the dynamic (or apparent) contact angle.}
\label{fig:post_dyn_contactangle}
\end{figure}
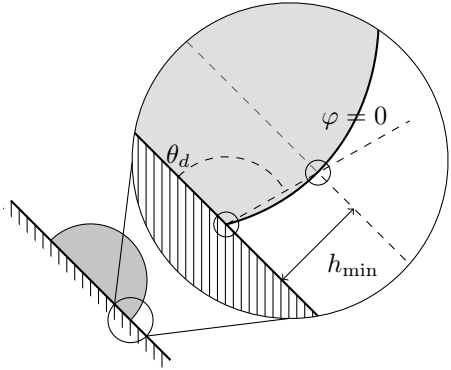

\subsubsection{Results}
\paragraph{Comparison of droplet shapes and characteristic values obtained on a coarse mesh and with a large time step}
In dependence on $\thetaori_s$, $r$ and $l$ the droplets show characteristic developments.
The calculated shapes for different combinations of $\thetaori_{s}$, $r$ and $l$ at
 $t=0.0; 0.5; 1.0; 1.5; 2.0$ are presented in \Cref{fig:sd_isolines}.
All the simulated droplets show the expected physical behavior: 
on the hydrophobic surface (third row) the droplet contracts, whereas the droplets spread on the hydrophilic surface (second row).
In addition, the droplets slide down the surface due to the density difference and gravity.
The different behavior at the advancing and receding contact points is visible and one can observe nonequilibrium contact angles in the second and third row. 
It is evident that there are virtually no differences between the coupled (solid black) and decoupled schemes (crosses) for all contact angles.
In contrast, in the linearized scheme with $\hp=100$ (dashed black) the dynamics are greatly reduced.
Similar as in the rising bubble case, a smaller $\hp$ ($\hp=10$, dashed gray) leads to improvements.
For comparison, we show the behavior of the droplet with the coupled scheme and $W^{poly_2}$ (dotted line).
Here, a slightly different droplet shape is observed especially for later times and large contact angles.

\begin{figure}
	\centering
	\vspace*{-3.5cm}
	\makebox[\textwidth][c]{
	\tikzsetnextfilename{sd_isolines}
	\begin{tikzpicture}
		\begin{groupplot}[,group style={
			rows=3, columns=4, horizontal sep=5pt, vertical sep=5pt
			,x descriptions at=edge bottom
			,y descriptions at=edge left 
			}
                ,width=2.0cm, height=8cm,
                ,grid=both
		,xmin=0, xmax=0.4,
		,ymin=0.4, ymax=2.0
		,scale only axis=true,
		,xlabel near ticks, ylabel near ticks
		,ylabel style={rotate=180}
		,xtick={0.1,0.3}
		,ytick={0.5,1.0,1.5}
            	]
\nextgroupplot[
			title={$t=0.5$}, 
			ylabel={$\thetaori_s=90\degree$, $r=0$, $l=1e6$}
			,legend to name=grouplegend2
			,legend columns=-1,
		] 
			\addplot[thick, smooth] plot file{data/slidingdroplet/isolines/obstacle100_coupled_nonlinear_0.001_0.0125_90.0_0.0_1000000.0_isolines_0.5.dat};
			\addplot[thick, mark=x, mark repeat=15, mark options={scale=2}] plot file {data/slidingdroplet/isolines/obstacle100_decoupled_nonlinear_0.001_0.0125_90.0_0.0_1000000.0_isolines_0.5.dat};
			\addplot[thick, dashed] plot file {data/slidingdroplet/isolines/obstacle100_decoupled_linear_0.001_0.0125_90.0_0.0_1000000.0_isolines_0.5.dat};
			\addplot[thick, mark=o, mark repeat=15] plot file {data/slidingdroplet/isolines/obstacle10_decoupled_linear_0.001_0.0125_90.0_0.0_1000000.0_isolines_0.5.dat};
			\addplot[thick, mark=*, mark repeat=15] plot file {data/slidingdroplet/isolines/poly2_coupled_nonlinear_0.001_0.0125_90.0_0.0_1000000.0_isolines_0.5.dat}; 
			\addlegendentry{coupl., $W^{\hp=100}$} 
			\addlegendentry{decoupl., $W^{\hp=100}$}
			\addlegendentry{decoupl./lin, $W^{\hp=100}$}
			\addlegendentry{decoupl./lin, $W^{\hp=10}$} 
			\addlegendentry{coupl., $W^{poly_2}$} 
		\nextgroupplot[
			title={$t=1.0$},
		]
			\addplot[thick, smooth] plot file{data/slidingdroplet/isolines/obstacle100_coupled_nonlinear_0.001_0.0125_90.0_0.0_1000000.0_isolines_1.dat};
			\addplot[thick, mark=x, mark repeat=15, mark options={scale=2}] plot file {data/slidingdroplet/isolines/obstacle100_decoupled_nonlinear_0.001_0.0125_90.0_0.0_1000000.0_isolines_1.dat};
			\addplot[thick, dashed] plot file {data/slidingdroplet/isolines/obstacle100_decoupled_linear_0.001_0.0125_90.0_0.0_1000000.0_isolines_1.dat};
			\addplot[thick, mark=o, mark repeat=15] plot file {data/slidingdroplet/isolines/obstacle10_decoupled_linear_0.001_0.0125_90.0_0.0_1000000.0_isolines_1.dat};
			\addplot[thick, mark=*, mark repeat=15] plot file {data/slidingdroplet/isolines/poly2_coupled_nonlinear_0.001_0.0125_90.0_0.0_1000000.0_isolines_1.dat};
		\nextgroupplot[
			title={$t=1.5$},
		]
			\addplot[thick, smooth] plot file{data/slidingdroplet/isolines/obstacle100_coupled_nonlinear_0.001_0.0125_90.0_0.0_1000000.0_isolines_1.5.dat};
			\addplot[thick, mark=x, mark repeat=15, mark options={scale=2}] plot file {data/slidingdroplet/isolines/obstacle100_decoupled_nonlinear_0.001_0.0125_90.0_0.0_1000000.0_isolines_1.5.dat};
			\addplot[thick, dashed] plot file {data/slidingdroplet/isolines/obstacle100_decoupled_linear_0.001_0.0125_90.0_0.0_1000000.0_isolines_1.5.dat};
			\addplot[thick, mark=o, mark repeat=15] plot file {data/slidingdroplet/isolines/obstacle10_decoupled_linear_0.001_0.0125_90.0_0.0_1000000.0_isolines_1.5.dat};
			\addplot[thick, mark=*, mark repeat=15] plot file {data/slidingdroplet/isolines/poly2_coupled_nonlinear_0.001_0.0125_90.0_0.0_1000000.0_isolines_1.5.dat}; 
		\nextgroupplot[
			title={$t=2.0$},
		]
			\addplot[thick, smooth] plot file{data/slidingdroplet/isolines/obstacle100_coupled_nonlinear_0.001_0.0125_90.0_0.0_1000000.0_isolines_2.dat};
			\addplot[thick, mark=x, mark repeat=15, mark options={scale=2}] plot file {data/slidingdroplet/isolines/obstacle100_decoupled_nonlinear_0.001_0.0125_90.0_0.0_1000000.0_isolines_2.dat};
			\addplot[thick, dashed] plot file {data/slidingdroplet/isolines/obstacle100_decoupled_linear_0.001_0.0125_90.0_0.0_1000000.0_isolines_2.dat};
			\addplot[thick, mark=o, mark repeat=15] plot file {data/slidingdroplet/isolines/obstacle10_decoupled_linear_0.001_0.0125_90.0_0.0_1000000.0_isolines_2.dat};
			\addplot[thick, mark=*, mark repeat=15] plot file {data/slidingdroplet/isolines/poly2_coupled_nonlinear_0.001_0.0125_90.0_0.0_1000000.0_isolines_2.dat}; 
\nextgroupplot[
			ylabel={$\thetaori_s=5\degree$, $r=0.35$, $l=140$}
		] 
			\addplot[thick, smooth] plot file{data/slidingdroplet/isolines/obstacle100_coupled_nonlinear_0.001_0.0125_5.0_0.35_140.0_isolines_0.5.dat};
			\addplot[thick, mark=x, mark repeat=15, mark options={scale=2}] plot file {data/slidingdroplet/isolines/obstacle100_decoupled_nonlinear_0.001_0.0125_5.0_0.35_140.0_isolines_0.5.dat};
			\addplot[thick, dashed] plot file {data/slidingdroplet/isolines/obstacle100_decoupled_linear_0.001_0.0125_5.0_0.35_140.0_isolines_0.5.dat};
			\addplot[thick, mark=o, mark repeat=15] plot file {data/slidingdroplet/isolines/obstacle10_decoupled_linear_0.001_0.0125_5.0_0.35_140.0_isolines_0.5.dat};
			\addplot[thick, mark=*, mark repeat=15] plot file {data/slidingdroplet/isolines/poly2_coupled_nonlinear_0.001_0.0125_5.0_0.35_140.0_isolines_0.5.dat}; 
		\nextgroupplot[
		]
			\addplot[thick, smooth] plot file{data/slidingdroplet/isolines/obstacle100_coupled_nonlinear_0.001_0.0125_5.0_0.35_140.0_isolines_1.dat};
			\addplot[thick, mark=x, mark repeat=15, mark options={scale=2}] plot file {data/slidingdroplet/isolines/obstacle100_decoupled_nonlinear_0.001_0.0125_5.0_0.35_140.0_isolines_1.dat};
			\addplot[thick, dashed] plot file {data/slidingdroplet/isolines/obstacle100_decoupled_linear_0.001_0.0125_5.0_0.35_140.0_isolines_1.dat};
			\addplot[thick, mark=o, mark repeat=15] plot file {data/slidingdroplet/isolines/obstacle10_decoupled_linear_0.001_0.0125_5.0_0.35_140.0_isolines_1.dat};
			\addplot[thick, mark=*, mark repeat=15] plot file {data/slidingdroplet/isolines/poly2_coupled_nonlinear_0.001_0.0125_5.0_0.35_140.0_isolines_1.dat};
		\nextgroupplot[
		]
			\addplot[thick, smooth] plot file{data/slidingdroplet/isolines/obstacle100_coupled_nonlinear_0.001_0.0125_5.0_0.35_140.0_isolines_1.5.dat};
			\addplot[thick, mark=x, mark repeat=15, mark options={scale=2}] plot file {data/slidingdroplet/isolines/obstacle100_decoupled_nonlinear_0.001_0.0125_5.0_0.35_140.0_isolines_1.5.dat};
			\addplot[thick, dashed] plot file {data/slidingdroplet/isolines/obstacle100_decoupled_linear_0.001_0.0125_5.0_0.35_140.0_isolines_1.5.dat};
			\addplot[thick, mark=o, mark repeat=15] plot file {data/slidingdroplet/isolines/obstacle10_decoupled_linear_0.001_0.0125_5.0_0.35_140.0_isolines_1.5.dat};
			\addplot[thick, mark=*, mark repeat=15] plot file {data/slidingdroplet/isolines/poly2_coupled_nonlinear_0.001_0.0125_5.0_0.35_140.0_isolines_1.5.dat}; 
		\nextgroupplot[
		]
			\addplot[thick, smooth] plot file{data/slidingdroplet/isolines/obstacle100_coupled_nonlinear_0.001_0.0125_5.0_0.35_140.0_isolines_2.dat};
			\addplot[thick, mark=x, mark repeat=15, mark options={scale=2}] plot file {data/slidingdroplet/isolines/obstacle100_decoupled_nonlinear_0.001_0.0125_5.0_0.35_140.0_isolines_2.dat};
			\addplot[thick, dashed] plot file {data/slidingdroplet/isolines/obstacle100_decoupled_linear_0.001_0.0125_5.0_0.35_140.0_isolines_2.dat};
			\addplot[thick, mark=o, mark repeat=15] plot file {data/slidingdroplet/isolines/obstacle10_decoupled_linear_0.001_0.0125_5.0_0.35_140.0_isolines_2.dat};
			\addplot[thick, mark=*, mark repeat=15] plot file {data/slidingdroplet/isolines/poly2_coupled_nonlinear_0.001_0.0125_5.0_0.35_140.0_isolines_2.dat};
\nextgroupplot[
			ylabel={$\thetaori_s=150\degree$, $r=0.35$, $l=140$}
		] 
			\addplot[thick, smooth] plot file{data/slidingdroplet/isolines/obstacle100_coupled_nonlinear_0.001_0.0125_150.0_0.35_140.0_isolines_0.5.dat};
			\addplot[thick, mark=x, mark repeat=15, mark options={scale=2}] plot file {data/slidingdroplet/isolines/obstacle100_decoupled_nonlinear_0.001_0.0125_150.0_0.35_140.0_isolines_0.5.dat};
			\addplot[thick, dashed] plot file {data/slidingdroplet/isolines/obstacle100_decoupled_linear_0.001_0.0125_150.0_0.35_140.0_isolines_0.5.dat};
			\addplot[thick, mark=o, mark repeat=15] plot file {data/slidingdroplet/isolines/obstacle10_decoupled_linear_0.001_0.0125_150.0_0.35_140.0_isolines_0.5.dat};
			\addplot[thick, mark=*, mark repeat=15] plot file {data/slidingdroplet/isolines/poly2_coupled_nonlinear_0.001_0.0125_150.0_0.35_140.0_isolines_0.5.dat}; 
		\nextgroupplot[
		]
			\addplot[thick, smooth] plot file{data/slidingdroplet/isolines/obstacle100_coupled_nonlinear_0.001_0.0125_150.0_0.35_140.0_isolines_1.dat};
			\addplot[thick, mark=x, mark repeat=15, mark options={scale=2}] plot file {data/slidingdroplet/isolines/obstacle100_decoupled_nonlinear_0.001_0.0125_150.0_0.35_140.0_isolines_1.dat};
			\addplot[thick, dashed] plot file {data/slidingdroplet/isolines/obstacle100_decoupled_linear_0.001_0.0125_150.0_0.35_140.0_isolines_1.dat};
			\addplot[thick, mark=o, mark repeat=15] plot file {data/slidingdroplet/isolines/obstacle10_decoupled_linear_0.001_0.0125_150.0_0.35_140.0_isolines_1.dat};
			\addplot[thick, mark=*, mark repeat=15] plot file {data/slidingdroplet/isolines/poly2_coupled_nonlinear_0.001_0.0125_150.0_0.35_140.0_isolines_1.dat};
		\nextgroupplot[
		]
			\addplot[thick, smooth] plot file{data/slidingdroplet/isolines/obstacle100_coupled_nonlinear_0.001_0.0125_150.0_0.35_140.0_isolines_1.5.dat};
			\addplot[thick, mark=x, mark repeat=15, mark options={scale=2}] plot file {data/slidingdroplet/isolines/obstacle100_decoupled_nonlinear_0.001_0.0125_150.0_0.35_140.0_isolines_1.5.dat};
			\addplot[thick, dashed] plot file {data/slidingdroplet/isolines/obstacle100_decoupled_linear_0.001_0.0125_150.0_0.35_140.0_isolines_1.5.dat};
			\addplot[thick, mark=o, mark repeat=15] plot file {data/slidingdroplet/isolines/obstacle10_decoupled_linear_0.001_0.0125_150.0_0.35_140.0_isolines_1.5.dat};
			\addplot[thick, mark=*, mark repeat=15] plot file {data/slidingdroplet/isolines/poly2_coupled_nonlinear_0.001_0.0125_150.0_0.35_140.0_isolines_1.5.dat}; 
		\nextgroupplot[
		]
			\addplot[thick, smooth] plot file{data/slidingdroplet/isolines/obstacle100_coupled_nonlinear_0.001_0.0125_150.0_0.35_140.0_isolines_2.dat};
			\addplot[thick, mark=x, mark repeat=15, mark options={scale=2}] plot file {data/slidingdroplet/isolines/obstacle100_decoupled_nonlinear_0.001_0.0125_150.0_0.35_140.0_isolines_2.dat};
			\addplot[thick, dashed] plot file {data/slidingdroplet/isolines/obstacle100_decoupled_linear_0.001_0.0125_150.0_0.35_140.0_isolines_2.dat};
			\addplot[thick, mark=o, mark repeat=15] plot file {data/slidingdroplet/isolines/obstacle10_decoupled_linear_0.001_0.0125_150.0_0.35_140.0_isolines_2.dat};
			\addplot[thick, mark=*, mark repeat=15] plot file {data/slidingdroplet/isolines/poly2_coupled_nonlinear_0.001_0.0125_150.0_0.35_140.0_isolines_2.dat};

	\end{groupplot}
	
    	\node (l1) at (group c1r2.west)
      		[rotate=-90,left, xshift=8cm,yshift=-1.7cm]
      		{\footnotesize \ref{grouplegend2}};

	\end{tikzpicture}
	}\caption{Shapes of the sliding droplets calculated with the schemes from \cref{sec:S}. Three different surfaces ranging from superhydrophilic ($5\degree$, middle) to superhydrophobic ($150\degree$, bottom) are compared. The corresponding parameters can be found in~\cref{tab:sd_setup,tab:sd_results}.}
	\label{fig:sd_isolines}
\end{figure}
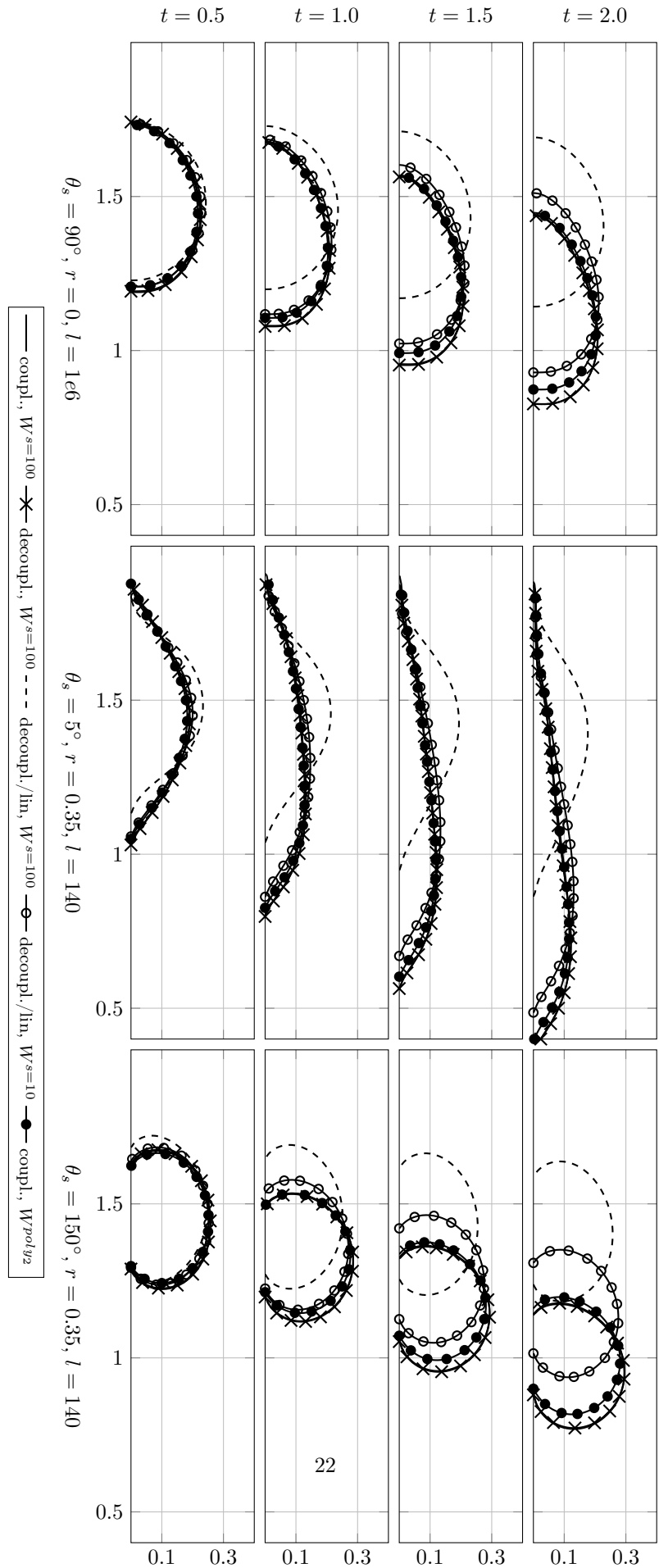
 
The evolution of the slide velocity $v_s$, the position of the contact points $y_p$ along the surface and the dynamic contact angle $\thetaori_d$ are displayed in \Cref{fig:sd_quantities_comparison}.
Again, we show all the schemes together with $W^{\hp=100}$ and in addition the stabilized scheme together with $W^{\hp=10}$ 
and the coupled, nonlinear scheme with $W^{poly_2}$.
To allow for a more quantitative comparison between the solution schemes, we list the characteristic values at $t=2$ in \Cref{tab:sd_results}.
As expected, in simulations without equilibrium contact angle relaxation ($r=0$) and slip ($l=1e6$) (first row) 
the apparent contact angles on both sides of the droplets stay near the equilibrium value $\thetaori_{s} = 90\degree$ the whole time.
In contrast, applying the full boundary conditions \cref{eq:M:6_NS_BC,eq:M:7_CH_BC} with $r=0.35$ and $l=140$ leads to clearly visible advancing and receding contact angles (third column). 
As before, no difference is visible between the coupled (solid black) and decoupled (black crosses) nonlinear schemes for all the characteristic quantities.
The characteristic values at $t=2$ differ only very slightly.
The results with the decoupled, stabilized scheme with $\hp=100$ (dashed black) 
are very far off and show very low sliding velocities (left column) 
and a different contact point behavior (middle column), especially for $\thetaori_{s}=150\degree$ (last row).
The sliding velocities at $t=2$ differ greatly.
In contrast to the comparison in the bulk only, see \Cref{ssec:N:rising_bubble}, 
the usage of the coupled scheme with $W^{poly_2}$ (dotted black) gives results which are noticeable different from the results with $W^{\hp=100}$.
This is most obvious in the simulations with $\thetaori_{s}=5\degree$ (middle row): 
the sliding velocity (left column) is slower and the terminal velocity is reached later. 
In addition, the receding and advancing contact angles are both lower than in the simulations with $W^{\hp=100}$.
For example, at $t=2$, the advancing contact angle for $W^{poly_2}$ is around $12\degree$ smaller than in the nonlinear simulations with $W^{\hp=100}$.
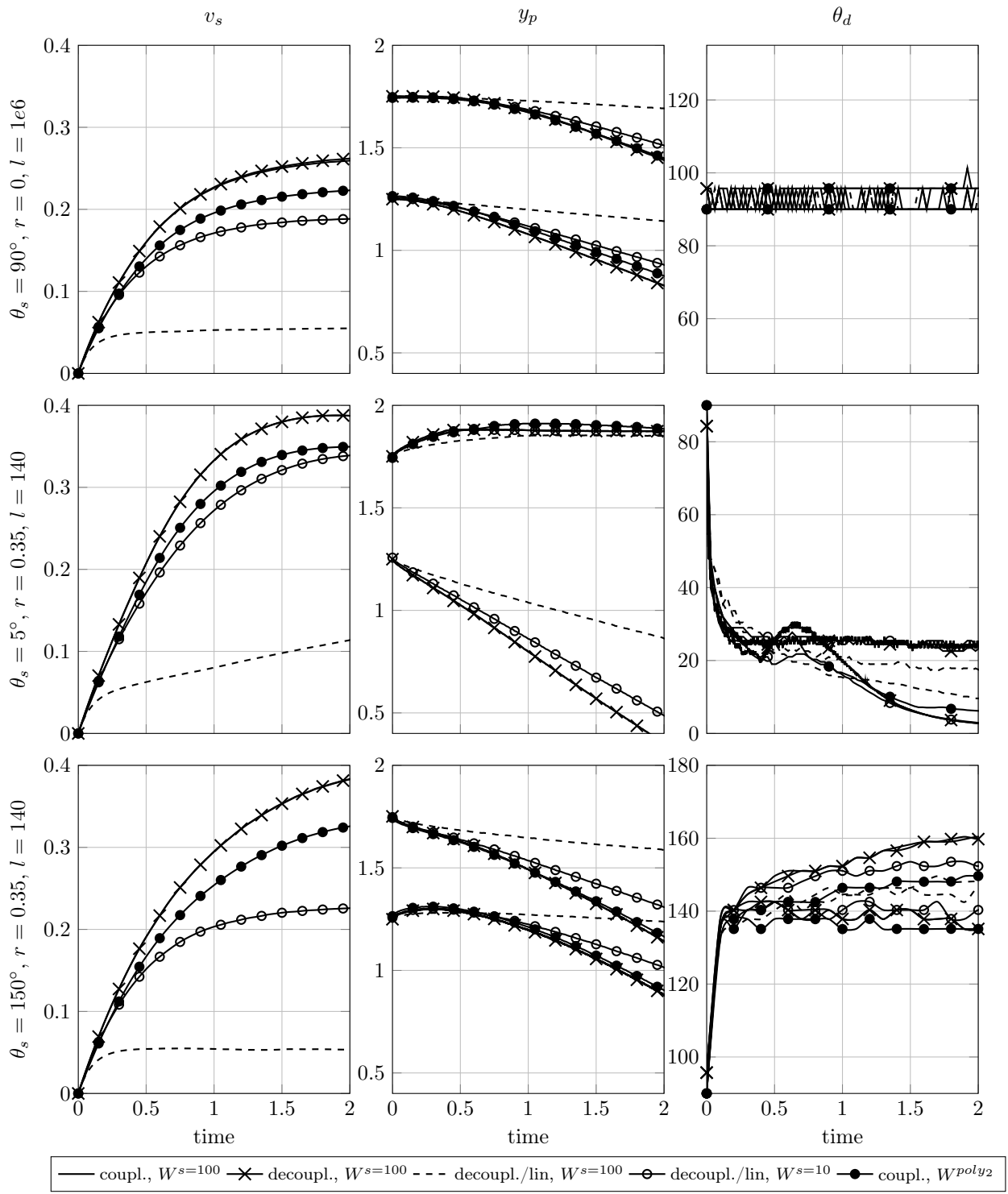
\begin{figure}
	\centering
\makebox[\textwidth][c]{
	\tikzsetnextfilename{sd_plots}
	\begin{tikzpicture}
		\begin{groupplot}[,group style={
			rows=3, columns=3, horizontal sep=20pt, vertical sep=15pt
			,x descriptions at=edge bottom
}
                ,width=0.50\textwidth
		,height=7cm
                ,grid=both
                ,xlabel={time} ,xmin=0, xmax=2
		,xlabel near ticks, ylabel near ticks
]
\nextgroupplot[
			title=$v_s$, 
			ylabel={$\thetaori_s=90\degree$, $r=0$, $l=1e6$}
			,legend to name=grouplegend
			,legend columns=-1
			,ymin=0, ymax=0.4
] 
			\addplot[thick, each nth point=10, filter discard warning=false, unbounded coords=discard] table[x=t,y expr=\thisrow{v_y}*(-1)] {data/slidingdroplet/data/obstacle100_coupled_nonlinear_0.001_0.0125_90.0_0.0_1000000.0_hysing.dat};
			\addplot[thick, each nth point=10, filter discard warning=false, unbounded coords=discard, mark=x, mark repeat=15, mark options={scale=2}] table[x=t,y expr=\thisrow{v_y}*(-1)] {data/slidingdroplet/data/obstacle100_decoupled_nonlinear_0.001_0.0125_90.0_0.0_1000000.0_hysing.dat};
			\addplot[thick, each nth point=10, filter discard warning=false, unbounded coords=discard, dashed] table[x=t,y expr=\thisrow{v_y}*(-1)] {data/slidingdroplet/data/obstacle100_decoupled_linear_0.001_0.0125_90.0_0.0_1000000.0_hysing.dat};
			\addplot[thick, each nth point=10, filter discard warning=false, unbounded coords=discard,mark=o, mark repeat=15] table[x=t,y expr=\thisrow{v_y}*(-1)] {data/slidingdroplet/data/obstacle10_decoupled_linear_0.001_0.0125_90.0_0.0_1000000.0_hysing.dat};
			\addplot[thick, each nth point=10, filter discard warning=false, unbounded coords=discard,mark=*, mark repeat=15] table[x=t,y expr=\thisrow{v_y}*(-1)] {data/slidingdroplet/data/poly2_coupled_nonlinear_0.001_0.0125_90.0_0.0_1000000.0_hysing.dat}; 
			\addlegendentry{coupl., $W^{\hp=100}$} 
			\addlegendentry{decoupl., $W^{\hp=100}$}
			\addlegendentry{decoupl./lin, $W^{\hp=100}$}
			\addlegendentry{decoupl./lin, $W^{\hp=10}$} 
			\addlegendentry{coupl., $W^{poly_2}$} 
		\nextgroupplot[
			title=$y_p$, 
,ymin=0.4, ymax=2.0
		] 
			\addplot[thick, each nth point=10, filter discard warning=false, unbounded coords=discard] table[x=t,y expr=\thisrow{y_p1}] {data/slidingdroplet/data/obstacle100_coupled_nonlinear_0.001_0.0125_90.0_0.0_1000000.0_contactline.dat};
			\addplot[thick, each nth point=10, filter discard warning=false, unbounded coords=discard, mark=x, mark repeat=15, mark options={scale=2}] table[x=t,y expr=\thisrow{y_p1}] {data/slidingdroplet/data/obstacle100_decoupled_nonlinear_0.001_0.0125_90.0_0.0_1000000.0_contactline.dat};
			\addplot[thick, each nth point=10, filter discard warning=false, unbounded coords=discard, dashed] table[x=t,y expr=\thisrow{y_p1}] {data/slidingdroplet/data/obstacle100_decoupled_linear_0.001_0.0125_90.0_0.0_1000000.0_contactline.dat};
			\addplot[thick, each nth point=10, filter discard warning=false, unbounded coords=discard,mark=o, mark repeat=15] table[x=t,y expr=\thisrow{y_p1}] {data/slidingdroplet/data/obstacle10_decoupled_linear_0.001_0.0125_90.0_0.0_1000000.0_contactline.dat};
			\addplot[thick, each nth point=10, filter discard warning=false, unbounded coords=discard,mark=*, mark repeat=15] table[x=t,y expr=\thisrow{y_p1}] {data/slidingdroplet/data/poly2_coupled_nonlinear_0.001_0.0125_90.0_0.0_1000000.0_contactline.dat};
			\addplot[thick, each nth point=10, filter discard warning=false, unbounded coords=discard] table[x=t,y expr=\thisrow{y_p2}] {data/slidingdroplet/data/obstacle100_coupled_nonlinear_0.001_0.0125_90.0_0.0_1000000.0_contactline.dat};
			\addplot[thick, each nth point=10, filter discard warning=false, unbounded coords=discard, mark=x, mark repeat=15, mark options={scale=2}] table[x=t,y expr=\thisrow{y_p2}] {data/slidingdroplet/data/obstacle100_decoupled_nonlinear_0.001_0.0125_90.0_0.0_1000000.0_contactline.dat};
			\addplot[thick, each nth point=10, filter discard warning=false, unbounded coords=discard, dashed] table[x=t,y expr=\thisrow{y_p2}] {data/slidingdroplet/data/obstacle100_decoupled_linear_0.001_0.0125_90.0_0.0_1000000.0_contactline.dat};
			\addplot[thick, each nth point=10, filter discard warning=false, unbounded coords=discard,mark=o, mark repeat=15] table[x=t,y expr=\thisrow{y_p2}] {data/slidingdroplet/data/obstacle10_decoupled_linear_0.001_0.0125_90.0_0.0_1000000.0_contactline.dat};
			\addplot[thick, each nth point=10, filter discard warning=false, unbounded coords=discard,mark=*, mark repeat=15] table[x=t,y expr=\thisrow{y_p2}] {data/slidingdroplet/data/poly2_coupled_nonlinear_0.001_0.0125_90.0_0.0_1000000.0_contactline.dat};		
		\nextgroupplot[
			title=$\thetaori_d$
			,ymin=45,ymax=135
] 
			\addplot[thick, each nth point=30, filter discard warning=false, unbounded coords=discard] table[x=t,y expr=\thisrow{theta_p1}] {data/slidingdroplet/data/obstacle100_coupled_nonlinear_0.001_0.0125_90.0_0.0_1000000.0_contactline.dat};
			\addplot[thick, each nth point=30, filter discard warning=false, unbounded coords=discard, mark=x, mark repeat=15, mark options={scale=2}] table[x=t,y expr=\thisrow{theta_p1}] {data/slidingdroplet/data/obstacle100_decoupled_nonlinear_0.001_0.0125_90.0_0.0_1000000.0_contactline.dat};
			\addplot[thick, each nth point=30, filter discard warning=false, unbounded coords=discard, dashed] table[x=t,y expr=\thisrow{theta_p1}] {data/slidingdroplet/data/obstacle100_decoupled_linear_0.001_0.0125_90.0_0.0_1000000.0_contactline.dat};
			\addplot[thick, each nth point=30, filter discard warning=false, unbounded coords=discard,mark=o, mark repeat=15] table[x=t,y expr=\thisrow{theta_p1}] {data/slidingdroplet/data/obstacle10_decoupled_linear_0.001_0.0125_90.0_0.0_1000000.0_contactline.dat};
			\addplot[thick, each nth point=30, filter discard warning=false, unbounded coords=discard,mark=*, mark repeat=15] table[x=t,y expr=\thisrow{theta_p1}] {data/slidingdroplet/data/poly2_coupled_nonlinear_0.001_0.0125_90.0_0.0_1000000.0_contactline.dat};
			\addplot[thick, each nth point=30, filter discard warning=false, unbounded coords=discard] table[x=t,y expr=\thisrow{theta_p2}] {data/slidingdroplet/data/obstacle100_coupled_nonlinear_0.001_0.0125_90.0_0.0_1000000.0_contactline.dat};
			\addplot[thick, each nth point=30, filter discard warning=false, unbounded coords=discard, mark=x, mark repeat=15, mark options={scale=2}] table[x=t,y expr=\thisrow{theta_p2}] {data/slidingdroplet/data/obstacle100_decoupled_nonlinear_0.001_0.0125_90.0_0.0_1000000.0_contactline.dat};
			\addplot[thick, each nth point=30, filter discard warning=false, unbounded coords=discard, dashed] table[x=t,y expr=\thisrow{theta_p2}] {data/slidingdroplet/data/obstacle100_decoupled_linear_0.001_0.0125_90.0_0.0_1000000.0_contactline.dat};
			\addplot[thick, each nth point=30, filter discard warning=false, unbounded coords=discard,mark=o, mark repeat=15] table[x=t,y expr=\thisrow{theta_p2}] {data/slidingdroplet/data/obstacle10_decoupled_linear_0.001_0.0125_90.0_0.0_1000000.0_contactline.dat};
			\addplot[thick, each nth point=30, filter discard warning=false, unbounded coords=discard,mark=*, mark repeat=15] table[x=t,y expr=\thisrow{theta_p2}] {data/slidingdroplet/data/poly2_coupled_nonlinear_0.001_0.0125_90.0_0.0_1000000.0_contactline.dat};	
\nextgroupplot[
ylabel={$\thetaori_s=5\degree$, $r=0.35$, $l=140$}
,ymin=0, ymax=0.4
		] 
			\addplot[thick, each nth point=10, filter discard warning=false, unbounded coords=discard] table[x=t,y expr=\thisrow{v_y}*(-1)] {data/slidingdroplet/data/obstacle100_coupled_nonlinear_0.001_0.0125_5.0_0.35_140.0_hysing.dat};
			\addplot[thick, each nth point=10, filter discard warning=false, unbounded coords=discard, mark=x, mark repeat=15, mark options={scale=2}] table[x=t,y expr=\thisrow{v_y}*(-1)] {data/slidingdroplet/data/obstacle100_decoupled_nonlinear_0.001_0.0125_5.0_0.35_140.0_hysing.dat};
			\addplot[thick, each nth point=10, filter discard warning=false, unbounded coords=discard, dashed] table[x=t,y expr=\thisrow{v_y}*(-1)] {data/slidingdroplet/data/obstacle100_decoupled_linear_0.001_0.0125_5.0_0.35_140.0_hysing.dat};
			\addplot[thick, each nth point=10, filter discard warning=false, unbounded coords=discard,mark=o, mark repeat=15] table[x=t,y expr=\thisrow{v_y}*(-1)] {data/slidingdroplet/data/obstacle10_decoupled_linear_0.001_0.0125_5.0_0.35_140.0_hysing.dat};
			\addplot[thick, each nth point=10, filter discard warning=false, unbounded coords=discard,mark=*, mark repeat=15] table[x=t,y expr=\thisrow{v_y}*(-1)] {data/slidingdroplet/data/poly2_coupled_nonlinear_0.001_0.0125_5.0_0.35_140.0_hysing.dat}; 
\nextgroupplot[
,ymin=0.4, ymax=2.0
		] 
			\addplot[thick, each nth point=10, filter discard warning=false, unbounded coords=discard] table[x=t,y expr=\thisrow{y_p1}] {data/slidingdroplet/data/obstacle100_coupled_nonlinear_0.001_0.0125_5.0_0.35_140.0_contactline.dat};
			\addplot[thick, each nth point=10, filter discard warning=false, unbounded coords=discard, mark=x, mark repeat=15, mark options={scale=2}] table[x=t,y expr=\thisrow{y_p1}] {data/slidingdroplet/data/obstacle100_decoupled_nonlinear_0.001_0.0125_5.0_0.35_140.0_contactline.dat};
			\addplot[thick, each nth point=10, filter discard warning=false, unbounded coords=discard, dashed] table[x=t,y expr=\thisrow{y_p1}] {data/slidingdroplet/data/obstacle100_decoupled_linear_0.001_0.0125_5.0_0.35_140.0_contactline.dat};
			\addplot[thick, each nth point=10, filter discard warning=false, unbounded coords=discard,mark=o, mark repeat=15] table[x=t,y expr=\thisrow{y_p1}] {data/slidingdroplet/data/obstacle10_decoupled_linear_0.001_0.0125_5.0_0.35_140.0_contactline.dat};
			\addplot[thick, each nth point=10, filter discard warning=false, unbounded coords=discard,mark=*, mark repeat=15] table[x=t,y expr=\thisrow{y_p1}] {data/slidingdroplet/data/poly2_coupled_nonlinear_0.001_0.0125_5.0_0.35_140.0_contactline.dat};
			\addplot[thick, each nth point=10, filter discard warning=false, unbounded coords=discard] table[x=t,y expr=\thisrow{y_p2}] {data/slidingdroplet/data/obstacle100_coupled_nonlinear_0.001_0.0125_5.0_0.35_140.0_contactline.dat};
			\addplot[thick, each nth point=10, filter discard warning=false, unbounded coords=discard, mark=x, mark repeat=15, mark options={scale=2}] table[x=t,y expr=\thisrow{y_p2}] {data/slidingdroplet/data/obstacle100_decoupled_nonlinear_0.001_0.0125_5.0_0.35_140.0_contactline.dat};
			\addplot[thick, each nth point=10, filter discard warning=false, unbounded coords=discard, dashed] table[x=t,y expr=\thisrow{y_p2}] {data/slidingdroplet/data/obstacle100_decoupled_linear_0.001_0.0125_5.0_0.35_140.0_contactline.dat};
			\addplot[thick, each nth point=10, filter discard warning=false, unbounded coords=discard,mark=o, mark repeat=15] table[x=t,y expr=\thisrow{y_p2}] {data/slidingdroplet/data/obstacle10_decoupled_linear_0.001_0.0125_5.0_0.35_140.0_contactline.dat};
\nextgroupplot[
ymin=0, ymax=90
		] 
			\addplot[thick, each nth point=1, filter discard warning=false, unbounded coords=discard] table[x=t,y expr=180-\thisrow{theta_p1}] {data/slidingdroplet/data/obstacle100_coupled_nonlinear_0.001_0.0125_5.0_0.35_140.0_contactline.dat};
			\addplot[thick, each nth point=30, filter discard warning=false, unbounded coords=discard, mark=x, mark repeat=15, mark options={scale=2}] table[x=t,y expr=180-\thisrow{theta_p1}] {data/slidingdroplet/data/obstacle100_decoupled_nonlinear_0.001_0.0125_5.0_0.35_140.0_contactline.dat};
			\addplot[thick, each nth point=30, filter discard warning=false, unbounded coords=discard, dashed] table[x=t,y expr=180-\thisrow{theta_p1}] {data/slidingdroplet/data/obstacle100_decoupled_linear_0.001_0.0125_5.0_0.35_140.0_contactline.dat};
			\addplot[thick, each nth point=30, filter discard warning=false, unbounded coords=discard,mark=o, mark repeat=15] table[x=t,y expr=180-\thisrow{theta_p1}] {data/slidingdroplet/data/obstacle10_decoupled_linear_0.001_0.0125_5.0_0.35_140.0_contactline.dat};
			\addplot[thick, each nth point=30, filter discard warning=false, unbounded coords=discard,mark=*, mark repeat=15] table[x=t,y expr=180-\thisrow{theta_p1}] {data/slidingdroplet/data/poly2_coupled_nonlinear_0.001_0.0125_5.0_0.35_140.0_contactline.dat};
			\addplot[thick, each nth point=1, filter discard warning=false, unbounded coords=discard] table[x=t,y expr=180-\thisrow{theta_p2}] {data/slidingdroplet/data/obstacle100_coupled_nonlinear_0.001_0.0125_5.0_0.35_140.0_contactline.dat};
			\addplot[thick, each nth point=30, filter discard warning=false, unbounded coords=discard, mark=x, mark repeat=15, mark options={scale=2}] table[x=t,y expr=180-\thisrow{theta_p2}] {data/slidingdroplet/data/obstacle100_decoupled_nonlinear_0.001_0.0125_5.0_0.35_140.0_contactline.dat};
			\addplot[thick, each nth point=30, filter discard warning=false, unbounded coords=discard, dashed] table[x=t,y expr=180-\thisrow{theta_p2}] {data/slidingdroplet/data/obstacle100_decoupled_linear_0.001_0.0125_5.0_0.35_140.0_contactline.dat};
			\addplot[thick, each nth point=30, filter discard warning=false, unbounded coords=discard,mark=o, mark repeat=15] table[x=t,y expr=180-\thisrow{theta_p2}] {data/slidingdroplet/data/obstacle10_decoupled_linear_0.001_0.0125_5.0_0.35_140.0_contactline.dat};
\nextgroupplot[
ylabel={$\thetaori_s=150\degree$, $r=0.35$, $l=140$}
,ymin=0, ymax=0.4
		] 
			\addplot[thick, each nth point=10, filter discard warning=false, unbounded coords=discard] table[x=t,y expr=\thisrow{v_y}*(-1)] {data/slidingdroplet/data/obstacle100_coupled_nonlinear_0.001_0.0125_150.0_0.35_140.0_hysing.dat};
			\addplot[thick, each nth point=10, filter discard warning=false, unbounded coords=discard, mark=x, mark repeat=15, mark options={scale=2}] table[x=t,y expr=\thisrow{v_y}*(-1)] {data/slidingdroplet/data/obstacle100_decoupled_nonlinear_0.001_0.0125_150.0_0.35_140.0_hysing.dat};
			\addplot[thick, each nth point=10, filter discard warning=false, unbounded coords=discard, dashed] table[x=t,y expr=\thisrow{v_y}*(-1)] {data/slidingdroplet/data/obstacle100_decoupled_linear_0.001_0.0125_150.0_0.35_140.0_hysing.dat};
			\addplot[thick, each nth point=10, filter discard warning=false, unbounded coords=discard,mark=o, mark repeat=15] table[x=t,y expr=\thisrow{v_y}*(-1)] {data/slidingdroplet/data/obstacle10_decoupled_linear_0.001_0.0125_150.0_0.35_140.0_hysing.dat};
			\addplot[thick, each nth point=10, filter discard warning=false, unbounded coords=discard,mark=*, mark repeat=15] table[x=t,y expr=\thisrow{v_y}*(-1)] {data/slidingdroplet/data/poly2_coupled_nonlinear_0.001_0.0125_150.0_0.35_140.0_hysing.dat}; 
\nextgroupplot[
,ymin=0.4, ymax=2.0
		] 
			\addplot[thick, each nth point=1, filter discard warning=false, unbounded coords=discard] table[x=t,y expr=\thisrow{y_p1}] {data/slidingdroplet/data/obstacle100_coupled_nonlinear_0.001_0.0125_150.0_0.35_140.0_contactline.dat};
			\addplot[thick, each nth point=10, filter discard warning=false, unbounded coords=discard, mark=x, mark repeat=15, mark options={scale=2}] table[x=t,y expr=\thisrow{y_p1}] {data/slidingdroplet/data/obstacle100_decoupled_nonlinear_0.001_0.0125_150.0_0.35_140.0_contactline.dat};
			\addplot[thick, each nth point=10, filter discard warning=false, unbounded coords=discard, dashed] table[x=t,y expr=\thisrow{y_p1}] {data/slidingdroplet/data/obstacle100_decoupled_linear_0.001_0.0125_150.0_0.35_140.0_contactline.dat};
			\addplot[thick, each nth point=10, filter discard warning=false, unbounded coords=discard,mark=o, mark repeat=15] table[x=t,y expr=\thisrow{y_p1}] {data/slidingdroplet/data/obstacle10_decoupled_linear_0.001_0.0125_150.0_0.35_140.0_contactline.dat};
			\addplot[thick, each nth point=10, filter discard warning=false, unbounded coords=discard,mark=*, mark repeat=15] table[x=t,y expr=\thisrow{y_p1}] {data/slidingdroplet/data/poly2_coupled_nonlinear_0.001_0.0125_150.0_0.35_140.0_contactline.dat};
			\addplot[thick, each nth point=1, filter discard warning=false, unbounded coords=discard] table[x=t,y expr=\thisrow{y_p2}] {data/slidingdroplet/data/obstacle100_coupled_nonlinear_0.001_0.0125_150.0_0.35_140.0_contactline.dat};
			\addplot[thick, each nth point=10, filter discard warning=false, unbounded coords=discard, mark=x, mark repeat=15, mark options={scale=2}] table[x=t,y expr=\thisrow{y_p2}] {data/slidingdroplet/data/obstacle100_decoupled_nonlinear_0.001_0.0125_150.0_0.35_140.0_contactline.dat};
			\addplot[thick, each nth point=10, filter discard warning=false, unbounded coords=discard, dashed] table[x=t,y expr=\thisrow{y_p2}] {data/slidingdroplet/data/obstacle100_decoupled_linear_0.001_0.0125_150.0_0.35_140.0_contactline.dat};
			\addplot[thick, each nth point=10, filter discard warning=false, unbounded coords=discard,mark=o, mark repeat=15] table[x=t,y expr=\thisrow{y_p2}] {data/slidingdroplet/data/obstacle10_decoupled_linear_0.001_0.0125_150.0_0.35_140.0_contactline.dat};
			\addplot[thick, each nth point=10, filter discard warning=false, unbounded coords=discard,mark=*, mark repeat=15] table[x=t,y expr=\thisrow{y_p2}] {data/slidingdroplet/data/poly2_coupled_nonlinear_0.001_0.0125_150.0_0.35_140.0_contactline.dat};		
		\nextgroupplot[
ymin=90, ymax=180
		] 
			\addplot[smooth, thick, each nth point=100, filter discard warning=false, unbounded coords=discard] table[x=t,y expr=\thisrow{theta_p1}] {data/slidingdroplet/data/obstacle100_coupled_nonlinear_0.001_0.0125_150.0_0.35_140.0_contactline.dat};
			\addplot[smooth, thick, each nth point=100, filter discard warning=false, unbounded coords=discard, mark=x, mark repeat=2, mark options={scale=2}] table[x=t,y expr=\thisrow{theta_p1}] {data/slidingdroplet/data/obstacle100_decoupled_nonlinear_0.001_0.0125_150.0_0.35_140.0_contactline.dat};
			\addplot[smooth, thick, each nth point=100, filter discard warning=false, unbounded coords=discard, dashed] table[x=t,y expr=\thisrow{theta_p1}] {data/slidingdroplet/data/obstacle100_decoupled_linear_0.001_0.0125_150.0_0.35_140.0_contactline.dat};
			\addplot[smooth, thick, each nth point=100, filter discard warning=false, unbounded coords=discard,mark=o, mark repeat=2] table[x=t,y expr=\thisrow{theta_p1}] {data/slidingdroplet/data/obstacle10_decoupled_linear_0.001_0.0125_150.0_0.35_140.0_contactline.dat};
			\addplot[smooth, thick, each nth point=100, filter discard warning=false, unbounded coords=discard,mark=*, mark repeat=2] table[x=t,y expr=\thisrow{theta_p1}] {data/slidingdroplet/data/poly2_coupled_nonlinear_0.001_0.0125_150.0_0.35_140.0_contactline.dat};
			\addplot[smooth, thick, each nth point=100, filter discard warning=false, unbounded coords=discard] table[x=t,y expr=\thisrow{theta_p2}] {data/slidingdroplet/data/obstacle100_coupled_nonlinear_0.001_0.0125_150.0_0.35_140.0_contactline.dat};
			\addplot[smooth, thick, each nth point=100, filter discard warning=false, unbounded coords=discard, mark=x, mark repeat=2, mark options={scale=2}] table[x=t,y expr=\thisrow{theta_p2}] {data/slidingdroplet/data/obstacle100_decoupled_nonlinear_0.001_0.0125_150.0_0.35_140.0_contactline.dat};
			\addplot[smooth, thick, each nth point=100, filter discard warning=false, unbounded coords=discard, dashed] table[x=t,y expr=\thisrow{theta_p2}] {data/slidingdroplet/data/obstacle100_decoupled_linear_0.001_0.0125_150.0_0.35_140.0_contactline.dat};
			\addplot[smooth, thick, each nth point=100, filter discard warning=false, unbounded coords=discard,mark=o, mark repeat=2] table[x=t,y expr=\thisrow{theta_p2}] {data/slidingdroplet/data/obstacle10_decoupled_linear_0.001_0.0125_150.0_0.35_140.0_contactline.dat};
			\addplot[smooth, thick, each nth point=100, filter discard warning=false, unbounded coords=discard,mark=*, mark repeat=2] table[x=t,y expr=\thisrow{theta_p2}] {data/slidingdroplet/data/poly2_coupled_nonlinear_0.001_0.0125_150.0_0.35_140.0_contactline.dat};
	\end{groupplot}
	
    	\node (l1) at ($(group c1r3.south)!0.5!(group c3r3.south)$)
      		[below, yshift=-2\pgfkeysvalueof{/pgfplots/every axis title shift}]
      		{\footnotesize \ref{grouplegend}};

	\end{tikzpicture}
	}\caption{Characteristic quantities calculated with the schemes from \Cref{sec:S}. 
	Three different surfaces ranging from super-hydrophilic ($5\degree$, middle) to super-hydrophobic ($150\degree$, bottom) are compared. The corresponding parameters can be found in~\cref{tab:sd_setup,tab:sd_results}}
	\label{fig:sd_quantities_comparison}
\end{figure}
 \begin{table}
\makebox[\textwidth][c]{
    	\csvreader[table head=\toprule Bulk pot.&Deco./Lin.?&$\tau$&$h$&$\thetaori$&$r$&$l$& $y_c$ &$y_{p,a}$&$y_{p,r}$&  $v_{s}$    & $c$	&$\thetaori_{d,a}$&$\thetaori_{d,r}$\\\midrule,
		table foot=\bottomrule,
		head to column names, 
  		late after last line=\\,
		before reading=\centering\sisetup{table-number-alignment=center, table-format=1.3,round-mode=places,round-precision=4},
		tabular={ll
			S[scientific-notation = true,table-format=1.0e1]
			S[table-format=0.5, round-precision=5]
			S[table-format=1.1]
			S[table-format=1.1]
			S[table-format=1.1]
			S[table-format=1.3, round-precision=4]
			S[table-format=1.3, round-precision=4]
			S[table-format=1.3, round-precision=4]
			S[table-format=1.3, round-precision=4]
			S[table-format=1.3, round-precision=4]
			S[table-format=2.0, round-precision=0]
			S[table-format=2.0, round-precision=0]
			},
		late after line=\ifnumequal{\thecsvrow}{8}{\\\midrule}{\\}
			\ifnumequal{\thecsvrow}{15}{\midrule}{}
		]{data/slidingdroplet/slidingdroplet.csv}{}{\pot&
			\decoupled/\linear&
			\ifnumequal{\thecsvrow}{1}{{\multirow{7}{*}{\dt}}}{}
			\ifnumequal{\thecsvrow}{8}{{\multirow{7}{*}{\dt}}}{}
			\ifnumequal{\thecsvrow}{15}{{\multirow{7}{*}{\dt}}}{}&
			\ifnumequal{\thecsvrow}{1}{{\multirow{7}{*}{\dx}}}{}
			\ifnumequal{\thecsvrow}{8}{{\multirow{7}{*}{\dx}}}{}
			\ifnumequal{\thecsvrow}{15}{{\multirow{7}{*}{\dx}}}{}&
			\ifnumequal{\thecsvrow}{1}{{\multirow{7}{*}{\theta}}}{}
			\ifnumequal{\thecsvrow}{8}{{\multirow{7}{*}{\theta}}}{}
			\ifnumequal{\thecsvrow}{15}{{\multirow{7}{*}{\theta}}}{}&
			\ifnumequal{\thecsvrow}{1}{{\multirow{7}{*}{\rel}}}{}
			\ifnumequal{\thecsvrow}{8}{{\multirow{7}{*}{\rel}}}{}
			\ifnumequal{\thecsvrow}{15}{{\multirow{7}{*}{\rel}}}{}&
			\ifnumequal{\thecsvrow}{1}{{\multirow{7}{*}{\slip}}}{}
			\ifnumequal{\thecsvrow}{8}{{\multirow{7}{*}{\slip}}}{}
			\ifnumequal{\thecsvrow}{15}{{\multirow{7}{*}{\slip}}}{}&
			\yct&
			\yca&
			\ycr&
			\vmax&
			\cmin&
			\thetaa&
			\thetar}
}
\caption{
		Parameters and characteristic values for the sliding droplets simulations. 
		For $y_c$ and $c$ see caption of \Cref{tab:rb_results}. 
		In addition, $y_{p}$ and $\thetaori_d$ denote the position of the contact points and the dynamic contact angles.
		The first and second values correspond to the advancing and receding contact point respectively angle. 
		The slide velocity is $v_s$ and all values are reported at $t=2$.
} 
\label{tab:sd_results}
\end{table}
 
\paragraph{Thermodynamic consistency and comparison of dissipation rates}
We reveal the thermodynamic consistency of the schemes by calculating the evolution of the energy inequality using \cref{eq:S:EnergyInequ}. 
We use $\gamma^{cc}$ and set $\thetaori_{s}=150\degree$, $r=0.35$ and $l=140$.
We observe, that $\Delta_n^m$ is positive for all times, which justifies that the schemes are thermodynamically consistent, see \cref{eq:S:EnergyInequ}. 
Note that for $r>0$ we introduce an additional error as soon as we use the decoupling strategy.
We observe, that the physical dissipation for the three nonlinear schemes are close together,	while the physical dissipation for the linear model is strongly reduced.
	This corresponds to the reduced dynamics that are observed in \Cref{fig:sd_isolines} for the linear schemes,	especially for $\hp=100$.
	This influence can be reduced by using very small time steps and finer meshes, see the results for the rising bubble case in \Cref{tab:rb_results}.
	Comparing the numerical and physical dissipation of the nonlinear schemes, $\Delta_n^m$ only accounts for around 25\% of the total dissipation even for large time steps.
	Furthermore, by halving the time step $\tau$, the numerical dissipation
	$\Delta_n^m$ relative to the total dissipation $\Delta_n^m+\Delta_p^m$ is greatly reduced to around 12\%, see the grey plots in the bottom figure of Figure~\ref{fig:sd_energy}.
\begin{figure}[!ht]
	\tikzsetnextfilename{energy_plots}
	\centering
	\begin{tikzpicture}
		\begin{groupplot}[,group style={
			rows=2, columns=1, vertical sep=14pt,
			,x descriptions at=edge bottom
			,y descriptions at=edge left 
			}
                ,width=0.60\textwidth
                ,height=6cm
                ,grid=both
                ,xlabel={time} ,ymin=0, ymax=2e-2,xmin=0, xmax=2
		,xlabel near ticks, ylabel near ticks
		,xtick={0.5,1.5}
		,ytick={0.5e-2, 1.5e-2}
]
\nextgroupplot[	
		ylabel=Physical diss. $\Delta_p^m$,  
		legend style={font=\footnotesize},legend pos=north west, 
		] 
			\addplot[thick, each nth point=10, filter discard warning=false, unbounded coords=discard] table[x=t,y=deltapwogravity] {data/slidingdroplet/energies/ws100_c_nlin.dat}; 
			\addplot[thick, each nth point=10, filter discard warning=false, unbounded coords=discard, color=black,mark=x, mark repeat=15, mark options={scale=2}] table[x=t,y=deltapwogravity] {data/slidingdroplet/energies/ws100_de_nlin.dat}; 
			\addplot[thick, each nth point=10, filter discard warning=false, unbounded coords=discard,dashed] table[x=t,y=deltapwogravity] {data/slidingdroplet/energies/ws100_de_lin.dat}; 
			\addplot[thick, each nth point=10, filter discard warning=false, unbounded coords=discard,mark=o, mark repeat=15] table[x=t,y=deltapwogravity] {data/slidingdroplet/energies/ws10_de_lin.dat}; 
			\addplot[thick, each nth point=10, filter discard warning=false, unbounded coords=discard,mark=*, mark repeat=15] table[x=t,y=deltapwogravity] {data/slidingdroplet/energies/poly2_c_nlin.dat}; 
			\addlegendentry{coupl., $W^{\hp=100}$} 
			\addlegendentry{deco., $W^{\hp=100}$}
			\addlegendentry{deco./lin, $W^{\hp=100}$}
			\addlegendentry{deco./lin, $W^{\hp=10}$} 
			\addlegendentry{coupl., $W^{poly_2}$} 
\nextgroupplot[
		ylabel=Numerical diss. $\Delta_n^m$, 
		] 
			\addplot[thick, each nth point=10, filter discard warning=false, unbounded coords=discard] table[x=t,y=deltan] {data/slidingdroplet/energies/ws100_c_nlin.dat}; 
			\addplot[thick, each nth point=10, filter discard warning=false, unbounded coords=discard, color=black,mark=x, mark repeat=15, mark options={scale=2}] table[x=t,y=deltan] {data/slidingdroplet/energies/ws100_de_nlin.dat}; 
			\addplot[thick, each nth point=10, filter discard warning=false, unbounded coords=discard,dashed] table[x=t,y=deltan] {data/slidingdroplet/energies/ws100_de_lin.dat}; 
			\addplot[thick, each nth point=10, filter discard warning=false, unbounded coords=discard,mark=o, mark repeat=15] table[x=t,y=deltan] {data/slidingdroplet/energies/ws10_de_lin.dat}; 
			\addplot[thick, each nth point=10, filter discard warning=false, unbounded coords=discard,mark=*, mark repeat=15] table[x=t,y=deltan] {data/slidingdroplet/energies/poly2_c_nlin.dat}; 
		\end{groupplot}
\begin{groupplot}[,group style={
			rows=2, columns=1, vertical sep=14pt,
			,x descriptions at=edge bottom
			,y descriptions at=edge right
			}
                ,width=0.60\textwidth
                ,height=6cm
                ,grid=none
		,ymin=0, ymax=0.25,xmin=0, xmax=2
		,xlabel near ticks, ylabel near ticks
		,xtick=\empty, axis line style=transparent,
]
\nextgroupplot[
			ytick=\empty
		] 
\nextgroupplot[
				ylabel=\textcolor{darkgray}{$\Delta_n^m/(\Delta_n^m+\Delta_p^m)$}
		] 
				\addplot[darkgray, thick, each nth point=90, filter discard warning=false, unbounded coords=discard] table[x=t,y expr=\thisrow{deltan}/(\thisrow{deltan}+\thisrow{deltapwogravity})] {data/slidingdroplet/energies/ws100_de_nlin.dat}; 
				\addplot[darkgray, thick, each nth point=90, filter discard warning=false, unbounded coords=discard] table[x=t,y expr=\thisrow{deltan}/(\thisrow{deltan}+\thisrow{deltap})] {data/slidingdroplet/energies/ws100_de_nlin_dt0.0005.dat}; 
				\node (a) at (axis cs:1.3,0.21){$\tau=0.001$};
				\node (b) at (axis cs:1.3,0.15){$\tau=0.0005$};
		\end{groupplot}
	\end{tikzpicture}
	\caption{Validity of the energy inequality \cref{eq:S:EnergyInequ} (bottom) and the physical dissipation \cref{eq:S:phyDiss} (top) for the different schemes. 
	All simulations are performed with $\tau=0.001$, $h_{min}=0.0125$, $\thetaori_{s}=150\degree$, $r=0.35$ and $l=140$. 
	The coupled/nonlinear and decoupled/nonlinear scheme match almost perfectly and appear as a single graph.
	The numerical dissipation relative to the total dissipation is shown for two different time steps sizes in the bottom figure. 
	}
\label{fig:sd_energy}
\end{figure}
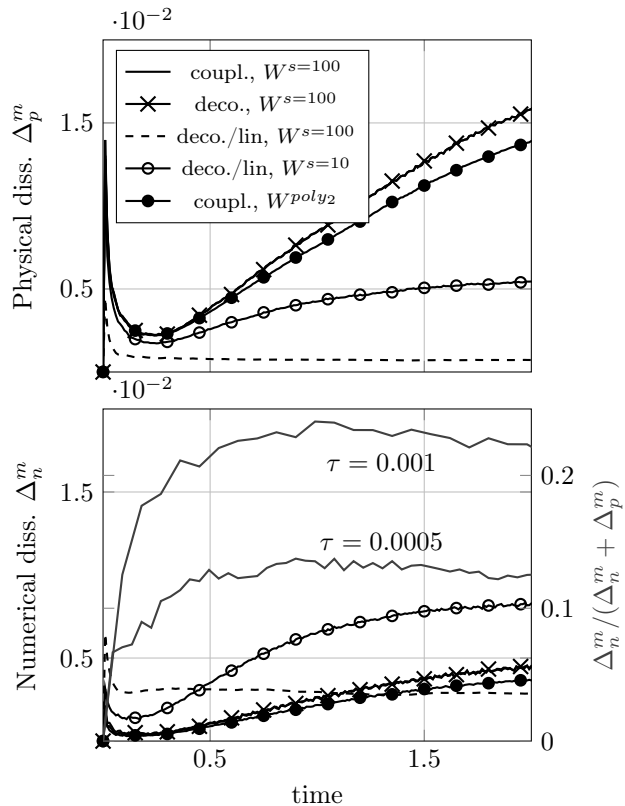 
 
\paragraph{Comparison of characteristic values obtained with smaller time steps $\tau$}
We show the behavior of the schemes for different time step sizes in Table~\ref{tab:sd_convergence}.
For small time steps, both the nonlinear schemes (coupled and decoupled) converge to the same characteristic values for the particular bulk energy potentials.
However, by comparing the values between the different bulk energy potentials, we note, that the differences are still relatively large even for small time steps.
Again, the linear scheme together with $W^\hp$ gives results far away from the solution obtained with the coupled schemes.
\begin{table}
	\csvreader[table head=\toprule Bulk pot.&Deco./Lin.?&$\tau$& $y_c$ &$y_{p,a}$&$y_{p,r}$&  $v_{s}$    & $c$	&$\thetaori_{d,a}$&$\thetaori_{d,r}$\\\midrule,
		table foot=\bottomrule,
		head to column names, 
		late after last line=\\,
		before reading=\centering\sisetup{table-number-alignment=center, table-format=1.3,round-mode=places,round-precision=4},
		tabular={ll
			S[scientific-notation = true,table-format=1.1e1, round-precision=1]
			S[table-format=1.3, round-precision=4]
			S[table-format=1.3, round-precision=4]
			S[table-format=1.3, round-precision=4]
			S[table-format=1.3, round-precision=4]
			S[table-format=1.3, round-precision=4]
			S[table-format=2.0, round-precision=0]
			S[table-format=2.0, round-precision=0]
			},
		late after line=\ifnumequal{\thecsvrow}{10}{\\\midrule}{\\}
			\ifnumequal{\thecsvrow}{13}{\midrule}{}
		]{data/slidingdroplet/convergence.csv}{}{\pot&
			\decoupled/\linear&
			\dt&
			\yct&
			\yca&
			\ycr&
			\vmax&
			\cmin&
			\thetaa&
			\thetar}
\caption{
		Characteristic values for the sliding droplet simulations obtained with different values of $\tau$ ($h_{min}=0.0125$, $\thetaori=150\degree$, $r=0.35$, $l=140$). For details about the setup and the characteristic values see the caption of \cref{tab:sd_results}.
}
\label{tab:sd_convergence}
\end{table}
 
\paragraph{Convergence to sharp-interface limit}
In \Cref{tab:sd_convergence_epsilon} we show solutions obtained with both bulk energy potentials and smaller $\epsilon$ on a very fine mesh ($h_{min}=0.0002$). 
	To reduce the computational effort, the inclination angle of the plate, see~\Cref{fig:sd_setup}, is set to zero and the simulation is already stopped at $t=0.2$.
	As expected for $b = \mathcal O(\epsilon)$ and $r = \mathcal O(1)$, see~\cite{2018-XuDiHu-SharpInterfaceLimit-NavierSlipBoundary}, the rate of convergence for both potentials is very slow and the sharp-interface limit is not reached yet.  
However, from our simulations we can conclude, that on the fine mesh both potentials give very similar results and exbibit the same behavior for smaller $\epsilon$. 
	For larger $\epsilon$, solutions obtained with $W^{\hp=100}$ seem to diverge slightly faster from the sharp-interface solution than solutions obtained with $W^{poly_2}$.
\begin{table}
	\csvreader[table head=\toprule 
			Bulk pot.&
			$\epsilon$&
			$b$&
			$y_p$
			\\\midrule,
		table foot=\bottomrule,
		head to column names, 
		late after last line=\\,
		before reading=\centering\sisetup{table-number-alignment=center, table-format=1.3,round-mode=places,round-precision=4},
		tabular={l
			S[table-format=0.2, round-precision=3]
			S[scientific-notation = true,table-format=1.0e1, round-precision=1]
			S[table-format=1.3, round-precision=4]
			},
		late after line=\ifnumequal{\thecsvrow}{5}{\\\midrule}{\\}
			\ifnumequal{\thecsvrow}{13}{\midrule}{}
		]{data/slidingdroplet/convergence_epsilon_recedingdrop.csv}{}{\potential&
			\eps&
			\mobility&
			\xcl
			}
\caption{
		Position of the contact line for a receding droplet (similar to the sliding droplet case with inclinication angle set to zero) obtained with different values of $\epsilon$ on a very fine mesh $h_{min}=0.0002$ ($\tau=0.001$, $\thetaori=150\degree$, $r=0.35$, $l=140$). The decoupled/nonlinear solution scheme is used. For details about the setup and the characteristic values see the caption of \cref{tab:sd_results}.
}
\label{tab:sd_convergence_epsilon}
\end{table}
   \section{Conclusion}
\label{sec:conclusion}
We compare the quality of the numerical results with three different schemes and two different bulk energy potentials.
For simulations without a moving contact line (rising bubble case), we find very similar results in the bulk independent of the coupling and linearization for both potentials.
However, the linearization of $W^\hp$ for large $\hp$ hinders the dynamics to a great extend but gets better for smaller $\hp$.
For the simulations including moving contact lines (sliding droplet case), the differences between the polynomial potential $W^{poly_2}$ and the relaxed double-obstacle potential $W^\hp$ are more pronounced.
Again, we observe a strong truncation of the allover dynamics using $W^\hp$ together with the linear scheme.
In both cases, the influence of the decoupling of the Navier--Stokes and Cahn--Hilliard system slightly depends on the time step size.
However, the decoupling has a negligible influence on the all-over dynamics even for larger time steps. 
Concerning the two tested bulk energy potentials, we observe, that both give in general physically sound results, but differences are still exists even for small time steps. 
The results and the behavior for smaller $\epsilon$ on a fine mesh are almost the same for both potentials.
For larger $\epsilon$, solutions obtained with $W^{\hp=100}$ seem to diverge slightly faster from a sharp-interface solution than with $W^{poly_2}$.

Summarizing our results, we find that
\begin{itemize}
	\item the decoupling strategy gives excellent results while the computational effort is significantly reduced compared to the fully coupled scheme,
	\item a further linearization of the Cahn--Hilliard system applying the stabilization is not recommended together with $W^\hp$ for large values of $\hp$, 
	\item 
	both bulk energy potentials produce sound and similar results in particular for a smaller interfacial thickness $\epsilon$.
\end{itemize}

To further judge whether one of the potentials lead to more accurate results, high fidelity sharp interface results on flows with moving contact lines (similar to the benchmark performed in~\cite{Hysing_Turek_quantitative_benchmark_computations_of_two_dimensional_bubble_dynamics}) are critical. It is our hope, that the presented work sparks further comparisons of diffuse and sharp interface models especially for the frequently observed and relevant case of sliding droplets.

\section*{Acknowledgment}
The authors thank Marion Dziwnik for helpful discussions on the scaling
of the contact line surface tensions. 

\bibliographystyle{elsarticle-num}

\end{document}